\documentclass[aip, jcp, amsmath, amssymb, reprint, floatfix]{revtex4-1}

\usepackage{graphicx}
\graphicspath{{figures/}}

\draft 

\usepackage{amsmath,amssymb,amsthm}
\usepackage{hyperref}
\usepackage{placeins}

\theoremstyle{definition}
\newtheorem{theorem}{Theorem}
\newtheorem{lemma}{Lemma}[section]
\newtheorem{proposition}[lemma]{Proposition}
\newtheorem{corollary}[lemma]{Corollary}
\newtheorem{definition}[lemma]{Definition}

\newtheoremstyle{myremark} 
{7pt}                    
{7pt}                    
{}  	                 
{}                           
{\bf}       	         
{.}                          
{.5em}                       
{}  
\theoremstyle{myremark}

\newcommand{\R}{\mathbb{R}}
\newcommand{\Z}{\mathbb{Z}}

\newcommand{\bdry}{\partial}
\newcommand{\M}[3]{\mathcal{M}^{#1}_{#2,#3}}
\newcommand{\cc}[1]{#1_\bullet}
\newcommand{\iso}{\cong}
\newcommand{\Tor}{\mathrm{Tor}}
\newcommand{\bcd}{\mathrm{bcd}}
\newcommand{\spnn}{\mathrm{span}}

\usepackage{xcolor}

\newcommand{\txrd}[1]{\textcolor{red}{#1}}              
\usepackage[normalem]{ulem}

\begin{document}


\title[Representations of Energy Landscapes by Sublevelset Persistent Homology]{Representations of Energy Landscapes by Sublevelset Persistent Homology: An Example With n-Alkanes}



\author{Joshua Mirth}
\affiliation{Department of Mathematics, Colorado State University, Fort Collins, Colorado 80524, USA}
\affiliation{Department of Computational Mathematics, Science, and Engineering, Michigan State University, East Lansing, Michigan 48824, USA}
\author{Yanqin Zhai}
\affiliation{Department of Nuclear, Plasma, and Radiological Engineering, University of Illinois at Urbana-Champaign, Urbana, Illinois 61801, USA}
\affiliation{Beckman Institute of Advanced Science and Technology, University of Illinois at Urbana-Champaign, Urbana, Illinois 61801, USA}
\author{Johnathan Bush}
\affiliation{Department of Mathematics, Colorado State University, Fort Collins, Colorado 80524, USA}
\author{Enrique G.\ Alvarado}
\affiliation{Department of Mathematics and Statistics, Washington State University, Pullman, Washington 99164, USA}
\author{Howie Jordan}
\affiliation{Department of Mathematics, University of Colorado, Boulder, Colorado 80309, USA}
\author{Mark Heim}
\affiliation{Department of Mathematics, Colorado State University, Fort Collins, Colorado 80524, USA}
\author{Bala Krishnamoorthy}
\affiliation{Department of Mathematics and Statistics, Washington State University, Vancouver, Washington 98686, USA}
\author{Markus Pflaum}
\affiliation{Department of Mathematics, University of Colorado, Boulder, Colorado 80309, USA}
\author{Aurora Clark}
\email[]{auclark@wsu.edu}
\affiliation{Department of Chemistry, Washington State University, Pullman, Washington 99164, USA}
\author{Y Z}
\email[]{zhyang@illinois.edu}
\affiliation{Department of Nuclear, Plasma, and Radiological Engineering, University of Illinois at Urbana-Champaign, Urbana, Illinois 61801, USA}
\affiliation{Beckman Institute of Advanced Science and Technology, University of Illinois at Urbana-Champaign, Urbana, Illinois 61801, USA}
\affiliation{Department of Electrical and Computer Engineering, University of Illinois at Urbana-Champaign, Urbana, Illinois 61801, USA}
\author{Henry Adams}
\email[]{henry.adams@colostate.edu}
\affiliation{Department of Mathematics, Colorado State University, Fort Collins, Colorado 80524, USA}


\date{\today}

\begin{abstract}
Encoding the complex features of an energy landscape is a challenging task, and often chemists pursue the most salient features (minima and barriers) along a highly reduced space, i.e.\ 2- or 3-dimensions.
Even though disconnectivity graphs or merge trees summarize the connectivity of the local minima of an energy landscape via the lowest-barrier pathways, there is more information to be gained by also considering the topology of each connected component at different energy thresholds (or sublevelsets).
We propose sublevelset persistent homology as an appropriate tool for this purpose.
Our computations on the configuration phase space of n-alkanes from butane to octane allow us to conjecture, and then prove, a complete characterization of the sublevelset persistent homology of the alkane $C_m H_{2m+2}$ Potential Energy Landscapes (PEL), for all $m$, and in all homological dimensions.
We further compare both the analytical configurational potential energy landscapes and sampled data from molecular dynamics simulation, using the united and all-atom descriptions of the intramolecular interactions.
In turn, this supports the application of distance metrics to quantify sampling fidelity and lays the foundation for future work regarding new metrics that quantify differences between the topological features of high-dimensional energy landscapes.

\smallskip
\noindent This manuscript was published in the Journal of Chemical Physics (\href{https://doi.org/10.1063/5.0036747}{DOI:10.1063/5.0036747}).
\end{abstract}

\pacs{}

\maketitle 



%
%

%

\section{Introduction}

High-dimensional energy landscapes (EL), including the configuration space of electrons, atoms, molecules, colloids, and other ``particles", frequently arise in materials science, chemistry, physics, and a wide range of dynamical systems.
One of the most common forms of an EL in chemical systems is the potential energy landscape (PEL) coming from the reduction of the electronic Hamiltonian in the Born--Oppenheimer approximation.\cite{BornOppenheimer,wales2003energy,stillinger2016ELs}
By this construction, the EL is a function across the nuclear configuration space of the system with dimension $3N$, where $N$ is the number of nuclei. 
How the physical system evolves is often determined by such an EL,\cite{Denzel1978, RevCompChem1997} but their high-dimensionality poses critical challenges to the analysis and understanding of them.
Indeed, one instance of the ``curse of dimensionality'' is that the space required to even store or represent such a surface grows exponentially with its dimension.
Further important variants of the EL may be based upon classical representations of inter-particle interactions, or may incorporate statistical mechanical ensembles of populated configurational states of the system to yield the free energy landscape.
The topological method used in this paper to describe ELs, namely sublevelset persistent homology, can be applied as well to all variants. 

Often chemists reduce the dimension of the EL from $3N$ to fewer by removing degrees of freedom that are not associated with the specific chemical phenomena of interest (e.g.\ rotations, translations, or the motion of atoms not of interest) according to chemical intuition and knowledge.
Conceptually it is convenient to reduce the dimensionality of the EL to two or three, yet in reality identifying the ``best set" of reduced variables, and understanding the extent of information loss upon dimensionality reduction, are significant and ongoing challenges.
One important criterion in regard to finding the ``best sets" of reduced variables is that topological descriptors such as the Morse indices of key critical points of chemical relevance of the EL survive the reduction process.
Although there are many examples where ``simple" geometric criteria suffice for reducing the landscape to represent a given chemical transformation, system complexity can cause coupling of many configurational degrees of freedom such that higher-dimensional representations of the EL are necessary.
Consider condensed phase ion-pairing reactions, where recent work has attempted to incorporate solvent reorganization by constructing ELs that depend upon ion separation distance, solvent density and solvent coordination.\cite{Mullen2014}
Alternatively, reduced variables that are intrinsically of dimensionality more than 1 have been employed (i.e.\ topological descriptors of intermolecular interactions).\cite{pietrucci2011graph,zhou2019pagerank}

The representation and visualization of surfaces described by scalar valued functions such as ELs is a highly nontrivial problem where the difficulty increases with the dimensionality of the surface.
It is thus attractive, and important, to consider how to represent the high-dimensional surface in a more compact form that could, in principle, support data-driven comparisons of ELs of different chemical systems that have different dimensionality.
Encoding or vectorizing complex chemical structures in phase space is already being employed in machine learning frameworks for materials discovery and design strategies.\cite{kulik2020, Butler2020}
Outside of highly local representations of ELs that are plotted and analyzed in 2--4 dimensions, a common representation of large and complex regions of ELs are disconnectivity graphs (or merge trees).\cite{Becker1997}
A disconnectivity graph compactly represents the energy value of each local minimum, and the energy barriers, usually the lowest barrier, required to pass between nearby local minima, which are critical to the chemistry of a system.
After their introduction in the late 1990s, the graphical properties of ELs have been exploited for a number of purposes, to both understand chemical transformations in multidimensional ELs (i.e.\ protein folding)\cite{Li2013} and for EL exploration.\cite{wales2003energy} 
A disconnectivity graph encodes the number of connected components in an EL; two configurations are in the same connected component if there is a path between the configurations that does not exceed the chosen energy barrier.
As the energy barrier increases, a disconnectivity graph stores how the connected components of the EL merge.
Formally, the disconnectivity graph of an EL can be understood as a tree graph with leaves corresponding to the local minima of the surface, and internal nodes to lowest energy values which are critical points of index one connecting the local minima.
The edges then correspond to pathways connecting the nodes within the surface below the given energy threshold. 
Disconnectivity graphs are widely and successfully used in the study of ELs,\cite{Wales2005} however, they do not capture all topological information of interest and also lack geometric information.
A related construction called lifted nearest neighbor graphs (NNGs) have been constructed using $0$-dimensional topological persistence to study local minima and index 1 saddle points of sampled energy landscapes.\cite{cazals2015conformational}
Metric disconnectivity graphs have been introduced to include some geometric information. \cite{SmeetonOakleyJohnston2014}
Reeb graphs are a refinement of disconnectivity graphs obtained by identifying in a surface described by a scalar function $f$ the points of a 
given level set $f^{-1} (c)$ which can be connected within that level set by a path. 
\cite{beketayev2011topology,reeb1946points}

Though disconnectivity graphs encode the energy barrier of the minimal energy pathway between two local minima, they do not represent multiple transition pathways.
This additional information is measured by sublevelset persistent homology. Described in more detail below, we briefly compare the essential features of ELs that this methodology can articulate.
Consider Figure~\ref{fig:MergeTreeVsPH2}, which shows an EL, transition pathways between local minima, and the energies associated to critical points of index 0 (minima), 1 (saddles), and 2.
The disconnectivity graph encodes the energy barriers associated with the local minima and the index 1 critical points (saddle points) that first merge local minima, and ignores all other transition paths.
In fact, 0-dimensional sublevelset persistent homology for the index 1 critical points has been previously employed to derive disconnectivity graphs and coarse grained representations of energy landscapes (based upon the energy filtration employed after including the effects of temperature).\cite{carr2016energy}
This information is reflected, however, in the 1-dimensional persistent homology.
Consider the two local minima on the top left, which have three transition paths between them.
The disconnectivity graph encodes that the local minima merge at energy level 2, but ignores the additional transition paths with energies 4 and 5 between them, which are important for the dynamics of the chemical system.
1-dimensional persistent homology (loops), however, measures the second transition pathway of energy 4, which creates a 1-dimensional bar that ends at energy level 6 when this pathway merges with the prior pathway of energy 2.
1-dimensional persistent homology also measures the third transition pathway of energy 5, which creates another 1-dimensional bar that ends at energy level 7 when this pathway merges with the prior pathway of energy 4.
In summary, whereas disconnectivity graphs only encode merge events between local minima, sublevelset persistent homology furthermore encodes merge events between transition pathways and higher-dimensional features.
Higher-dimensional persistent homology barcodes encode relative barrier heights between $k$-dimensional features, for all $k$, enabling more accurate estimations of the time scale of the associated chemical dynamics.

\begin{figure}[ht]
\centering
\includegraphics[width=0.45\textwidth]{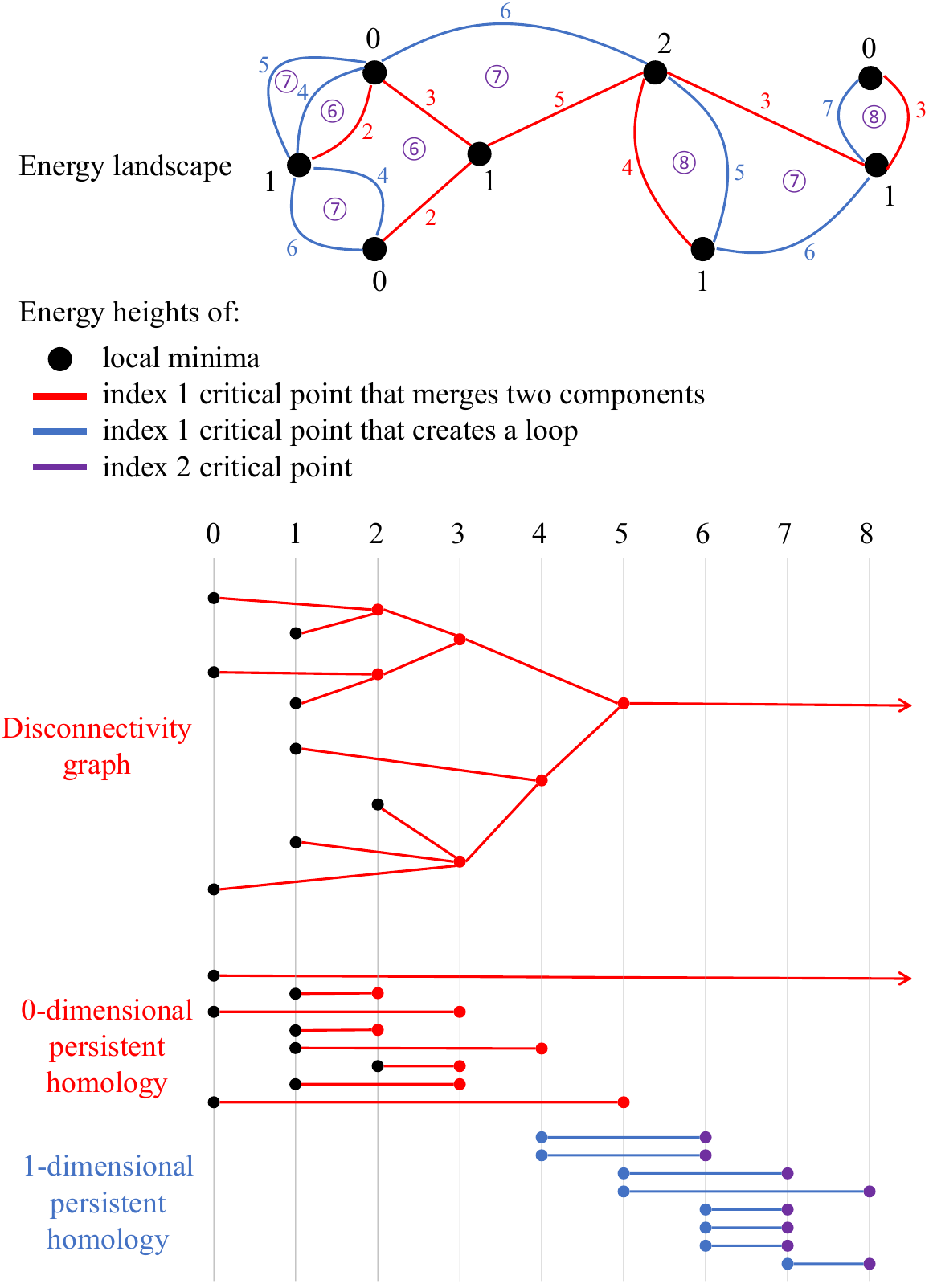}
\caption{
A map of an energy landscape, its disconnectivity graph, and its persistent homology.
The red edges are the edges of the minimal spanning tree (merging) connected components; the blue edges instead create 1-dimensional loops.
The 0-dimensional persistent homology is obtained by simply cutting ``joints" in the disconnectivity graph, and then laying the branches flat.
Whereas the disconnectivity graph or 0-dimensional persistent homology encodes the local minima and index 1 critical points merging components, the 1-dimensional persistent homology also encodes the index 1 critical points whose transition pathways form loops, along with the energy barriers of index 2 critical points that fill between two transition pathways.
}
\label{fig:MergeTreeVsPH2}
\end{figure}

Sublevelset persistent homology derives from persistent homology, a technique that has been used in the chemistry community to summarize the shapes of a molecule, supramolecular assemblies, and other complex structures that emerge in multicomponent solutions.
Given a set of atom locations in $\R^3$, one can build a simplicial complex (consisting of vertices, edges, triangles, tetrahedra) whose vertex locations are given by the positions of the atoms.
Arguably the first paper on persistent homology\cite{edelsbrunner2000topological} contains computations on the nonribosomal peptide gramicidin A, on a portion of a periodic zeolite, and on a portion of DNA.
Perhaps the first simplicial complexes built on top of molecules from the perspective of quantifying topology were alpha complexes,\cite{edelsbrunner1994three,edelsbrunner2005geometry} which pre-date and in fact helped lead to the invention of persistent homology.
The philosophy of persistence, and indeed multi-parameter persistence, has been used to improve the robustness of clustering (disconnectivity graphs) of molecular dynamics data~\cite{chang2013persistent}, though this summary is only of 0-dimensional homology.
The nudged elastic band method\cite{henkelman2002methods} has also inspired new methods in topological data analysis (TDA).\cite{NEBinTDA}
More modern applications of persistent homology to chemistry include Refs.~\citenum{xia2015persistent,xia2014persistent,townsend2020representation}, which better account for differences between different atom identities, and which successfully use the persistent homology barcodes for machine learning tasks.
By contrast, instead of considering the persistent homology of a single molecular configuration, Ref.~\citenum{membrillo2019topology} considers the topology of the entire configurational space, either on its own or alternatively equipped with an energy function, including the case of pentane.
This work is heavily inspired by Ref.~\citenum{martin2010topology}, which studies the energy landscape of the cyclo-octane molecule.

Herein, we propose that  sublevelset persistent homology is an efficient quantitative descriptor of high-dimensional surfaces, as demonstrated by the study of a series of energy landscapes, with increasing dimensions. 
In the process, we create a dictionary between the lower-dimensional topological features measured by disconnectivity graphs that are also measured by sublevelset persistent homology.
These low-dimensional features are furthermore generalized by higher-dimensional features measured by persistent homology but not disconnectivity graphs.
As an example system, the Potential Energy Landscape (PEL) of n-alkanes, from n-butane to n-octane, is considered and analyzed in this paper.
(For the remainder of this paper, we refer to the n-alkanes simply as alkanes, as these are the only alkanes we consider.)
Within this series:
the collective variables for the reduced EL (based upon the 4 carbon-center dihedral angles) are known;
reducing dimensions down to the collective variables preserves the critical points of index 0;
the dimensionality of the EL is systematically increased by adding more carbons to the chain;
the relative barrier heights in the EL can be varied by considering bonded vs.\ non-bonded interatomic interactions.
Of particular note is that the EL of butane can be embedded in pentane, which can be embedded in hexane, etc, which then makes it amenable to derive and prove the sublevelset persistent homologies of all alkane chains.
The alkanes also exhibit surprisingly rich behavior in their physicochemical properties that is, in part, related to the configurational EL.
For example, the odd-even effect of alkanes, namely, the zig-zag variations of their physicochemical properties such as the melting points, solid densities, sublimation enthalpy, solubility, and modulus as a function of the number of carbon atoms, has been known since 1877. \cite{baeyer1877ueber, badea2006odd, mishra2013odd, boese1999melting}
However, only recently, a dynamic odd-even effect in the transport properties of liquid alkanes was discovered,\cite{yang2016dynamic} challenging the established understanding of the odd-even effect as a consequence of packing efficiency in the alkanes crystalline solids.\cite{boese1999melting}
Therefore, the topology of the EL may provide a new perspective to elucidate the intriguing physicochemical properties of alkanes.
Toward that end, our computations on the alkanes up through octane allow us to conjecture, and then prove, a complete characterization of the sublevelset persistent homology of the alkane C$_m$H$_{2m+2}$ PEL, for all $m$, and in all homological dimensions.

\section{Sublevelset Persistent Homology}
\label{sec:sublevelset-ph}

\subsection{Introduction of Sublevelset Persistent Homology}

Given a real-valued function defined on some potentially high-dimensional domain, {\it sublevelset persistent homology} is a way to visualize the shape of that function, and in particular, the shape of its various sublevelsets.
Much of the popularity of persistent homology\cite{Carlsson2009,EdelsbrunnerHarer,PersistentImages} stems from the fact that it is computable.
Indeed, sublevelset persistent homology can be computed with a running time that is sub-cubic in the number of cells (vertices, edges, triangles or squares, tetrahedra or cubes, etc.) needed to mesh the domain.\cite{milosavljevic2011zigzag}

Let $X$ be a space and let $f\colon X\to \R$ be a real-valued function.
For example, $X=S^1\times S^1$ may be a torus which encodes the two dihedral angles in a pentane molecule configuration, and $f\colon X\to \R$ may be the reduced EL for pentane, ignoring the hydrogen degrees of freedom and the vibrational and bending modes of carbons.
The {\it sublevelset of $f \colon X \to \R$ at height $r$} is $f^{-1}((-\infty,r])=\{x\in X~|~f(x)\le r\}$.
That is, the sublevelset is the set of all points in $X$ whose value under $f$ is at most $r$.
In the case of the pentane EL, this sublevelset encodes the set of all conformations of the pentane molecule with energy at most $r$.
In other words, the sublevelset $f^{-1}((-\infty,r])$ is the restricted energy landscape, with energy restricted to be at most $r$.
As $r$ varies from small to large, sublevelset persistent homology describes how the topology of the sublevelsets change.

For $r\le r'$, note that we have an inclusion of sublevelsets $f^{-1}((-\infty,r])\hookrightarrow f^{-1}((-\infty,r'])$.
As a result, we can track how the homology (the number of holes in each dimension) changes as the value of $r$ increases.
0-dimensional holes correspond to connected components, 1-dimensional holes correspond to loops, and 2-dimensional holes correspond to voids, etc.
Persistent homology allows us to count the number of holes in a restricted energy landscape $f^{-1}((-\infty,r])$ not only for a single value of $r$, but also over an evolution of energies as $r$ increases from small to large.

\subsection{Advantages of Sublevelset Persistent Homology Barcodes}
\label{sec:results-dictionary}

Energy Landscapes are often described by disconnectivity graphs,\cite{wales2003energy} or equivalently merge trees, that encode how connected components of the EL appear and merge as the energy threshold is raised.
A disconnectivity graph analysis, however, ignores the topology of connected components: a connected component that is a disk is treated the same as a connected component that has holes or voids.
We propose the use of sublevelset persistent homology to understand more of the topology of restricted energy landscapes, including not only 0-dimensional persistent homology (which is closely related to the disconnectivity graph), but also higher-dimensional persistent homology.
Whereas the disconnectivity graph depends only on local minima and critical points of index 1 in the energy landscape, the sublevelset persistent homology depends on critical points of all indices (including minima and critical points of index 1 as special cases).
Indeed, if the EL is a ``Morse function,"
then each $k$-dimensional bar in the persistent homology barcode has a birth time corresponding to the energy value of a critical point of index $k$, i.e.\ a critical point at which there are $k$ linearly independent directions in which the energy value decreases.
Furthermore, a death time of a bar in the sublevelset persistent homology of a Morse function corresponds to the energy value of a critical point of index $k+1$.
A Morse function, defined rigorously in Appendix~\ref{app:num-bars}, is a smooth real-valued function with no degenerate critical points.

\begin{figure}[ht]
\includegraphics[width=\columnwidth]{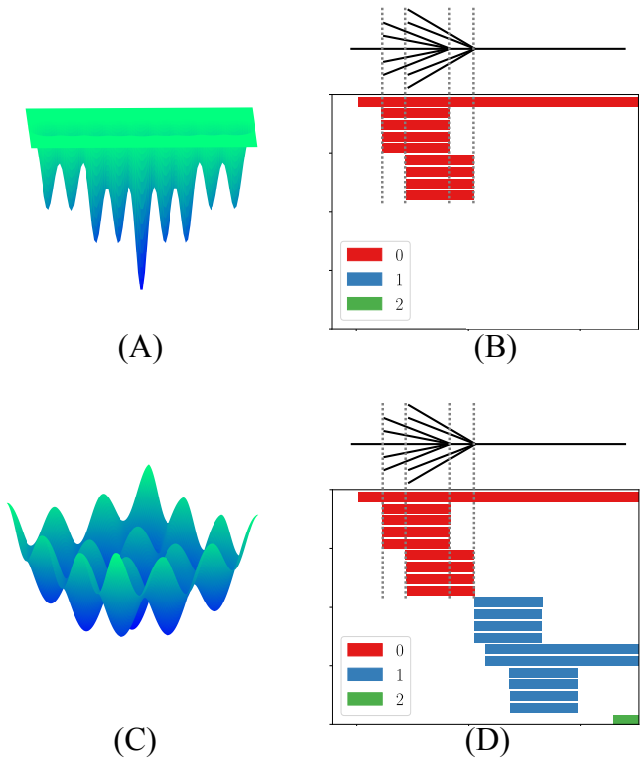}
\caption{
While a disconnectivity graph summarizes the connected components in an energy landscape, it does not portray the topology of each connected component.
The two ELs drawn above have the same disconnectivity graph.
However, the connected components of sublevelsets of the EL (A) are contractible, whereas those of the EL (C) may contain holes, as shown by the sublevelsets in Figure~\ref{fig:pentane-sublevelsets}.
The 0-dimensional persistent homology in red measures much of the same information as the disconnectivity graph, but the 1-dimensional persistent homology in blue distinguishes these two ELs.
The $y$-axis in (A) and (C) is energy,
as is the $x$-axis in (B) and (D).
The $y$-axis in (B) and (D) is an ordering and count of the bars.
For this and for all of the remaining figures in the paper, we order the bars vertically by birth energy, but any vertical reordering of the bars gives a valid representation of the persistent homology.
}
\label{fig:merge-vs-PH}
\end{figure}

As a motivating example, consider Figure~\ref{fig:merge-vs-PH}.
We see two different ELs whose disconnectivity graphs are exactly the same.
This similarity is also reflected in the 0-dimensional sublevelset persistent homology barcode plots (red bars)---indeed, the disconnectivity graph of an EL completely determines its 0-dimensional sublevelset persistent homology.
Nevertheless, the two energy landscapes are quite different.
In the EL shown in Figure~\ref{fig:merge-vs-PH}A, all connected components of sublevelsets are connected.
In the EL shown in Figure~\ref{fig:merge-vs-PH}C, the connected components instead have a ``Swiss cheese" structure, with a variety of holes.
These differences are reflected in the 1-dimensional persistent homology sublevelsets (blue bars).
One should think of sublevelset persistent homology as a higher-dimensional analogue of disconnectivity graphs, tracking higher-dimensional spatial features of an EL.
Indeed, sublevelset persistent homology has the power to distinguish ELs whose disconnectivity graphs are identical, and to provide summaries of the higher-dimensional topology of the ELs.

Although it has been traditionally held that that the presence of higher-dimensional topology and its associated transition pathways (e.g., second-order saddles) are inconsequential to the mechanisms and dynamics of chemical processes, their presence in energy landscapes is well-established.\cite{Pradhan2017, Dmitrenko2004, Shaik2001}  Further, recent work has demonstrated from \emph{ab-initio} molecular dynamics that some reactions can follow second order pathways, for example the denitrogenation of 1-pyrazoline,\cite{Pradhan2019}  wherein the statistical observation of the second-order path is likely altered by available energy.
Though disconnectivity graphs encode all of the local minima of an EL, they retain only some of the index 1 critical points, i.e.\ transition states between local minima.
By contrast, from the sublevelset persistent homology one can compute the critical points of all indices, including all \emph{transition states} between local minima, as explained in Section~\ref{sec:ph-alkane}.
Each additional transition pathway provides an alternative relaxation mechanism of the system that is missing from the disconnectivity graph.
Therefore, the timescale computed from a disconnectivity graph alone will be an overestimate to a certain degree.
Persistent homology measures not only all transition pathways, but also the relative barrier heights between them (see Figure~\ref{fig:MergeTreeVsPH2}), enabling more rigorous statistical mechanics of the dynamics of the system.

Another important feature of persistent homology is that it is {\it stable}, meaning that small changes to the input produce only small changes in the persistent homology.\cite{cohen2007stability}
This continuity property is necessary in any data analysis technique, as input noise or measurement error is unavoidable.
In Section~\ref{ssec:ua-aa} we quantify the difference between the analytical and OPLS-UA alkane energy landscapes using this notion of similarity between persistent homology barcodes.

Another powerful topological descriptor that can be used to visualize the surface described by a real-valued function is the Morse--Smale complex.\txrd{\cite{BanyagaHurtubise}}
It can be associated to any EL described by a potential energy function which fulfills the transversality condition; see Appendix~\ref{app:morse-complex}.
In contrast to filtering the EL via sublevelsets, the Morse--Smale complex provides a partitioning of the EL into pieces whose points have equal behavior in regard to the gradient flow of the potential energy function.
As a result, the Morse--Smale complex captures crucial information chemists are interested in, namely identifying and describing areas within the EL with similar energetic behavior such as points which are attracted to the same local minimum.
Moreover, it contains strictly more topological information than the sublevelset persistent homology, such as how critical points of adjacent indices are connected to each other along gradient paths.
The increase in topological-geometric information comes with the price that Morse--Smale complexes can only be visualized easily for surfaces of dimension at most three.
We can compute the Morse--Smale complex for some small chain-length alkanes as seen in Figure~\ref{fig:EL-butane-pentane}.

\section{Physical Data Sets and Methods}

\subsection{The Analytical PEL of Alkanes}
\label{ssec:analytical}

In this paper, we choose gas phase alkanes as a model system to demonstrate the application of sublevelset persistent homology.
In the absence of intermolecular interactions, the intramolecular interactions were described using the the Optimized Potentials for Liquid Simulations (OPLS) force field, which, in general, consists of energies of bonds, angles, dihedrals, and nonbonded interactions.\cite{jorgensen1988opls} 
We first adopted the OPLS-UA (united atom) approach by coarse-graining the hydrogen degrees of freedom into the parameters of adjacent carbon atoms implicitly.
To further simplify the EL, we fixed all bond lengths and the three-body angles.
The nonbonded intramolecular interactions were also ignored.
Therefore, the potential energy landscape (PEL) of a single alkane molecule is only governed by the C--C--C--C dihedral angles $\phi_i$ as
\begin{align}
\label{eq:alkane-pel}
&V(\phi_1, \ldots, \phi_{m-3})=\nonumber\\
&\sum_{i=1}^{m-3}\left(c_1\left[1+\cos\phi_i\right]+c_2\left[1-\cos2\phi_i\right]+c_3\left[1+\cos3\phi_i\right]\right),
\end{align}
where $m$ is the number of carbon atoms, the energy coefficients are $c_1/k_B=355.03$~K, $c_2/k_B=-68.19$~K, and $c_3/k_B=791.32$~K, and $k_B$ is the Boltzmann constant.\cite{chen1998thermodynamic}
The PEL of an alkane molecule with $m$ carbon atoms has a dimension of $n=m-3$ in this simplistic analytical model, in the sense that the PEL is a real-valued function $V=f_n \colon (S^1)^n\to\R$, where each circular factor $S^1$ encodes a dihedral angle $\phi_i$.
The reduction from $\R^{3m}$ down to this lower-dimensional PEL on $(S^1)^n$ satisfies the important property that the indices of index zero critical points are preserved.
An observation that will be useful in our classification of the sublevelset persistent homology of the alkane PEL is that
\begin{align*}
&f_n(\phi_1, \ldots, \phi_n)=V(\phi_1, \ldots, \phi_n)\\
=&V(\phi_1)+\ldots+V(\phi_n)=f_1(\phi_1)+\ldots+f_1(\phi_n).
\end{align*}

\begin{figure}
\begin{center}
\includegraphics[width=0.8\columnwidth]{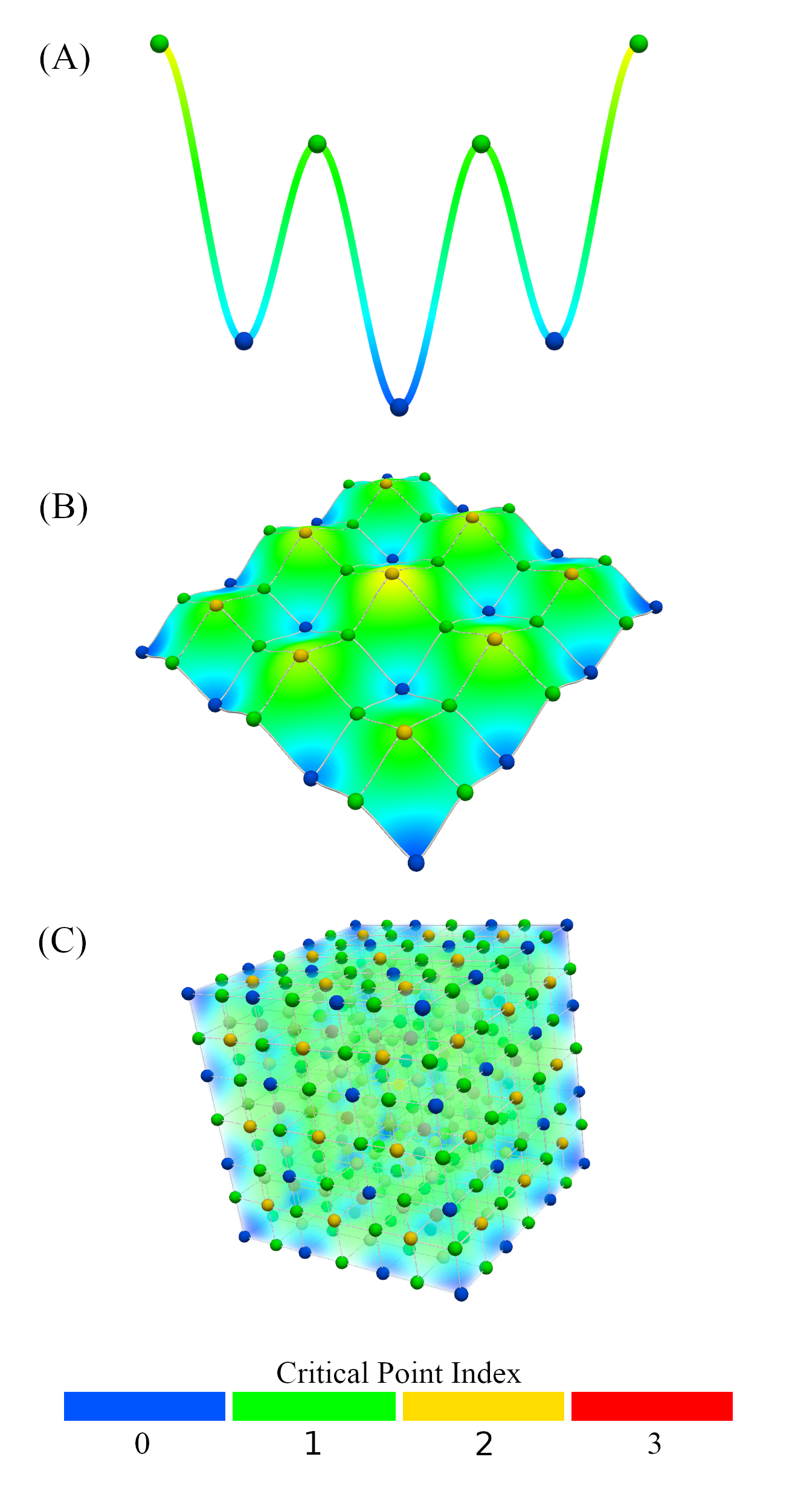}
\end{center}
\caption{
Morse--Smale complexes for the analytical PEL for (A) butane, (B) pentane, and (C) hexane.
The critical points are indicated and are colored by their index.
For pentane and hexane, the unique flows between critical points whose indices differ by one are also indicated; see Appendix~\ref{app:morse-smale}.
For (A) and (B), the energy scales are indicated by the vertical direction.}
\label{fig:EL-butane-pentane}
\end{figure}

\subsection{Molecular Dynamics Simulation of the Coarse-Grained and All-Atom Models of Alkanes}

The second PEL data set is obtained from molecular dynamics (MD) simulations of the alkane OPLS-UA model.
A single alkane molecule, from butane to hexane ($m$ from 4 to 6), was placed in the center of a 10 nm cubic box with periodic boundary conditions enforced in all three dimensions.
All bonds and angles were fixed using the SHAKE algorithm.\cite{ryckaert1977numerical}
The MD simulations were performed in the canonical ensemble using the Nos\'{e}-Hoover thermostat, where the temperature was was set to $RT\approx831$~kJ/mol and employed to overcome the potential energy barriers of the system and ensure ergodic sampling.
Initially, the velocity was randomly assigned to each atom according to the Maxwell--Boltzmann distribution.
The integration time was 0.1~fs and the molecular configuration and potential energy were collected every 0.1~ps for 1,000,000 frames.
The resulting data set consists of the sampled configurational potential energy landscape for all dihedral rotations in the system.
All the simulations were performed using GROMACS 5.0.7.\cite{abraham2015gromacs}

The PEL obtained from the united atom approach was then compared to our third PEL, the all-atom (OPLS-AA) model.
We used the same simulation protocol for OPLS-AA as we did with OPLS-UA except that:
(i) the C--C--C angles were left flexible to ensure that the total number of constraints does not exceed the number of degrees of freedom,
(ii) the temperature was set to $RT\approx16.6$~kJ/mol to reduce the effect of the flexible angles on the total potential energy, and
(iii) the molecular configuration and potential energy were collected every 1~fs for 10,000,000 frames to obtain better statistics.
In this model, all hydrogen atoms are considered explicitly. 
Consequently, extra dimensions associated with the hydrogen degrees of freedom were introduced to the PEL. 
In order to compare with the OPLS-UA model, the potential energy corresponding to a specific dihedral angle configuration was calculated by averaging all the potential energies sampled at that configuration, i.e.\ we averaged the contribution of the hydrogens to the PEL and the reduced the PEL to be only as a function of the C--C--C--C dihedral angles.
The resulting PEL is qualitatively similar to the PEL obtained with the OPLS-UA model.
However, the averaging produced a rougher PEL with small local features, and the energy values of all features have been shifted higher.

\subsection{Computation of Sublevelset Persistent Homology}

We computed the sublevelset persistent homology of the PEL of alkanes up through octane using the GUDHI software package;\cite{maria2014gudhi} our code is publicly available.\cite{delta-persistence}
For the analytical OPLS-UA, we computed the sublevelset persistent homology of a cubical grid with 63 (approximately $10\cdot 2\pi$) vertices on each circular axis.
This grid was treated as a filtered cubical complex via the lower-star filtration, assigning to higher-dimensional cubes a filtration value matching the maximum energy value on a boundary vertex.
For the MD simulations of $1,000,000$ configurations for OPLS-UA we computed the sublevelset persistent homology as follows.
First we downsampled to $1,000$ vertices using sequential maxmin,\cite{de2004topological} and then computed the sublevelset lower-star filtration of a Delaunay triangulation with periodic boundary conditions.
The size of the resulting complexes are given in Table~\ref{table:delaunay}.

\begin{table}[htp!]
    \centering
\begin{tabular}{|l|r|r|} \hline
Dimension & UA-Pentane & UA-Hexane  \\ \hline
0     & $1000$ & $ 1000$ \\ \hline
1 & $3000$ & $7528$ \\ \hline
2 & $2000$ & $13056$ \\ \hline
3 & 0 & $6528$ \\ \hline
Total & $6000$ & $28112$ \\ \hline
\end{tabular}
    \caption{Number of simplices in the periodic Delaunay triangulation for UA-Pentane and UA-Hexane.}
    \label{table:delaunay}
\end{table}

For the MD simulation of $1,000,000$ configurations of OPLS-AA we subdivided the domain into a cubical grid with $100$ cubical regions on each circular axis. The energy values were then averaged over $9\times9$ patches of cubes, as described in Section~\ref{ssec:ua-aa}, and the persistent homology computed of the lower-star filtration of the cubical grid with respect to the averaged values.

\section{Sublevelset Persistent Homology of the PEL of Alkanes}
\label{sec:ph-alkane}

\subsection{Sublevelset Persistent Homology of the Analytical PEL of Pentane}

\begin{figure*}
\centering
\includegraphics[width=1\textwidth]{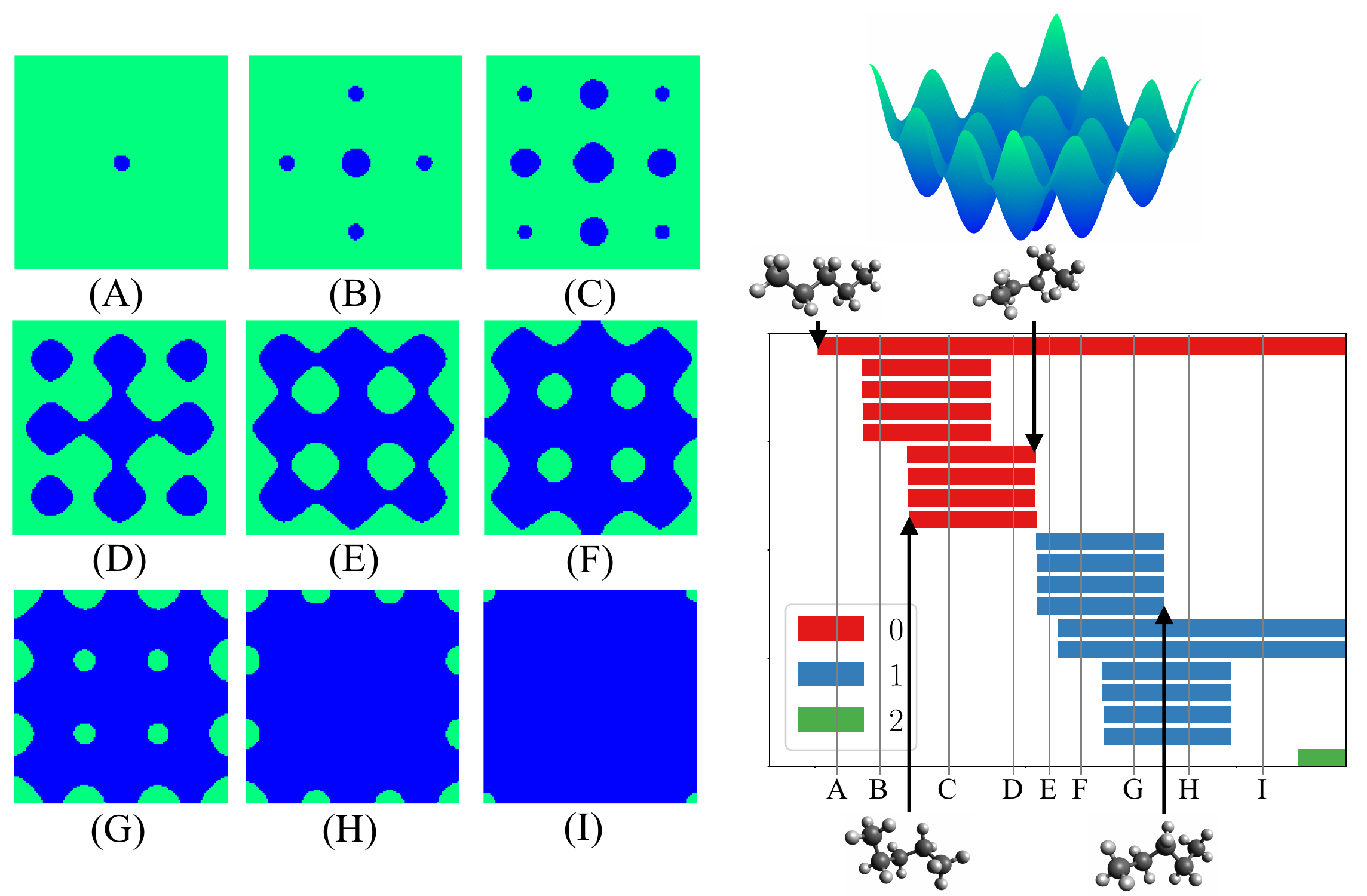}
\caption{
(Left) Pentane sublevelsets $f_2^{-1}(-\infty,r]:=\{y\in (S^1)^2~|~f_2(y)\le r\}$ are drawn, in green, for increasing values of energy value $r$.
(Right) The pentane sublevelset persistent homology of the analytical PEL.
Red bars are 0-dimensional features (connected components), blue bars are 1-dimensional features (loops), and the green bar is the lone 2-dimensional feature (the entire torus).
The $x$-axis is energy (kJ/mol).
}
\label{fig:pentane-sublevelsets}
\end{figure*}

We consider now the sublevelset persistent homology of the analytical PEL of pentane in Figure~\ref{fig:pentane-sublevelsets}.
Sublevelset (A) contains a single connected component, corresponding to the first 0-dimensional bar that is born in the persistent homology.
This configuration represents the global energy minimum of pentane, i.e.\ a fully stretched pentane. 
In sublevelset (B), four new connected components form, and then four more in (C), giving 8 new 0-dimensional persistent homology bars.
(B) and (C) represent two types of local minima configurations of pentane, with corresponding energies of 3.58 kJ/mol and 7.15 kJ/mol respectively.
Upon including four saddle points of index one, in (D) we reduce back down to four connected components.
The saddle points are configurations with $\phi_1 = \pi$ and $\phi_2 = \pi\pm2\pi/3$, or alternatively $\phi_1 = \pi\pm2\pi/3$ and $\phi_2 = \pi$, and an energy of 13.78 kJ/mol.
Upon passing eight more saddle points of index one, in (E) we obtain a connected space with four 1-dimensional holes.
We see that a critical point of index 1, i.e.\ a saddle point, can either correspond to the energy barrier connecting nearby local minima (these are the saddle points found, for example, by nudged elastic band), or it could alternatively correspond to the birth of a new 1-dimensional hole as measured by persistent homology.
As the sublevelsets continue to grow, a variety of other 1-dimensional holes are born.
In the transition from (G) to (H), we include four critical points of index 2, which each kill a 1-dimensional persistent homology bar.
These configurations are local maxima of pentane, with $\phi_1 = \pi\pm\pi/3$ and $\phi_2 = \pi\pm\pi/3$ and an energy of 27.57 kJ/mol, as shown in Figure~\ref{fig:pentane-sublevelsets}.
The final sublevelset (not shown) is the entire 2-dimensional torus, which has a single connected component, two 1-dimensional holes, and a single 2-dimensional hole.
The topological features in the final sublevelset correspond to the {\it semi-infinite bars} as they have no death times; all other bars are {\it finite}.

Now we explain how all critical points of all indices are obtainable from the sublevelset persistent homology computation.
Each local minimum, i.e.\ critical point of index 0, is represented by the birth of a 0-dimensional bar.
The death of a 0-dimensional bar corresponds to a critical point of index 1.
The above are the only critical points measured by disconnectivity graphs, but persistent homology encodes the remaining critical points of index 1 as births of 1-dimensional bars.
For example, in the transition from (D) to (E), eight new saddle points of index 1 are included, corresponding to four 0-dimensional bars ending (reducing from 5 connected components down to 1), and to four 1-dimensional bars that are born.
More generally, the number of critical points of index $k$ in the energy landscape is equal to the number of $k$-dimensional persistent homology bars, plus the number of finite (i.e., not semi-infinite) $(k-1)$-dimensional bars.
The energy values of these index $k$ critical points are given by the birth time of the corresponding $k$-dimensional bar, or alternatively the death time of the corresponding $(k-1)$-dimensional bar.

The length of each persistent homology bar measures the prominence of each topological feature.
This is clear for 0-dimensional bars; the prominence of local minimum is measured in the same way as in disconnectivity graphs.
However, persistent homology also measures the prominence of 1-dimensional features; see Figure~\ref{fig:MergeTreeVsPH2}.

\subsection{From Butane to Octane}
\label{ssec:butane-octane}

The analytical PEL for the alkanes are defined as follows.
We index the alkane C$_m$H$_{2m+2}$ by the number of dihedral angles $n=m-3$.
Each PEL is a function $V=f_n\colon (S^1)^n\to\R$ given in Equation~\eqref{eq:alkane-pel}, where $S^1$ is the circle encoding a dihedral angle, and where $(S^1)^n$ is the $n$-dimensional torus.
The PEL for butane is given by a real-valued function $f_1\colon S^1\to \R$ defined on the circle, shown in Figure~\ref{fig:EL-butane-pentane}A.
We see that $f_1$ has three local minima, of which one is global, and three local maxima, with the single global maxima pictured twice due to the periodic boundary conditions.
An equivalent definition of the alkane energy landscape function $f_n \colon (S^1)^n\to \R$ is given by $f_n(\phi_1, \ldots, \phi_n) =  f_1(\phi_1) + \ldots + f_1(\phi_n)$.
See Figure~\ref{fig:alkane-barcodes} for the sublevelset persistent homology barcodes of the analytical PELs of butane through octane.

\begin{figure*}
\centering
\includegraphics[width=1\textwidth]{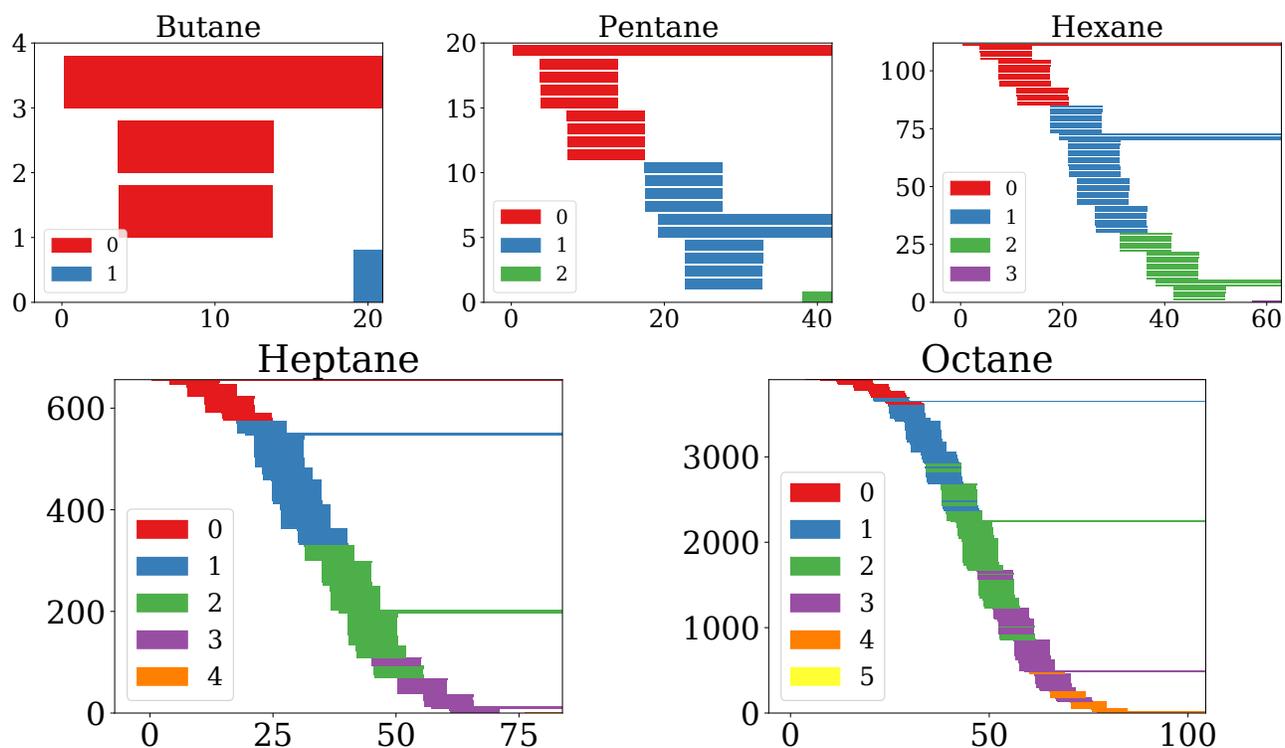}
\caption{
Persistent homology barcodes for the alkane analytical PEL, from butane ($n=1$) through octane ($n=5$).
The color of the bar indicates the homological dimension.
The $y$-axis counts the number of bars.
The $x$-axis is energy (kJ/mol).
}
\label{fig:alkane-barcodes}
\end{figure*}

\subsection{Sublevelset Persistent Homology of the PEL of Alkanes From MD Simulation}
\label{ssec:ua-aa}

\begin{figure*}
\centering
\includegraphics[width=1\textwidth]{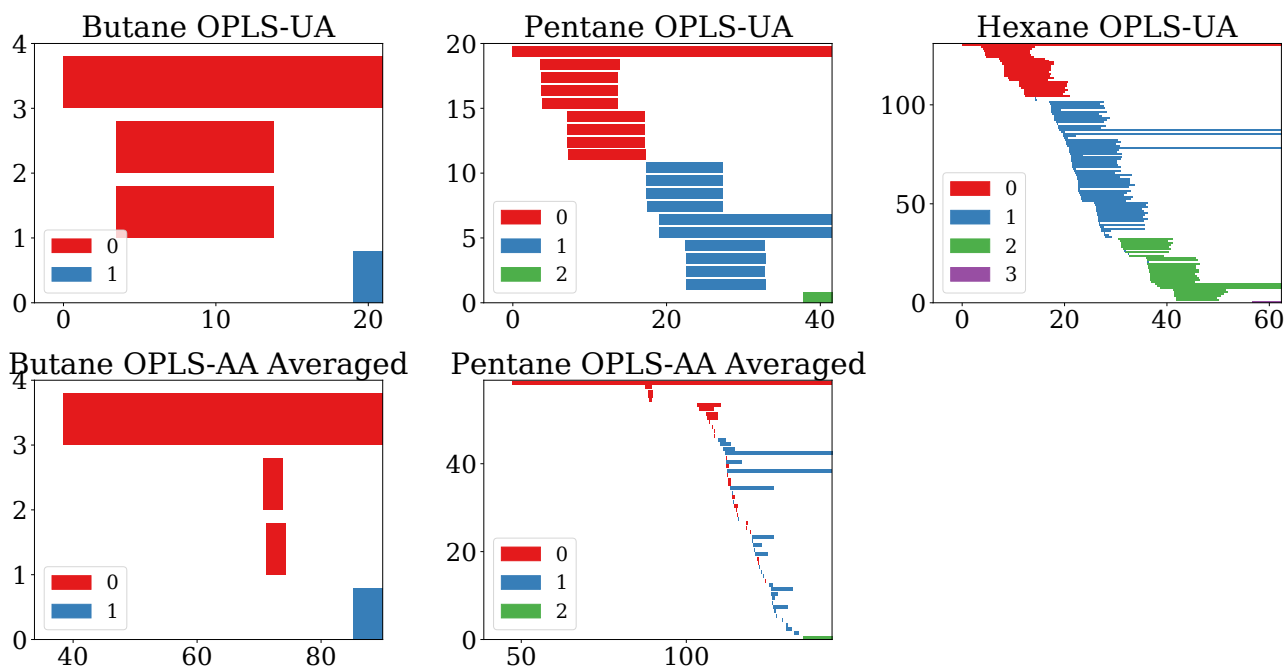}
\caption{
Sublevelset persistent homology figures with butane, pentane, hexane in the columns, and with the OPLS-UA and OPLS-AA PELs from MD simulations in the rows.
The $y$-axis counts the number of bars.
The $x$-axis is energy (kJ/mol).
We do not include the barcodes for hexane OPLS-AA because the H--C--H and H--C--C angle fluctuations cause too much noise in the PEL compared to the MD sampling density.
}
\label{fig:3x3}
\end{figure*}

To compute the sublevelset persistent homology of the MD data, specifically the energies of sampled configurations, we first computed Delaunay triangulations on the finite sample of the $n$-dimensional torus $(S^1)^n$.
Our triangulation is formed by quotienting out Delaunay triangulations computed on Euclidean space, avoiding the complexities of Delaunay triangulations in arbitrary Riemannian manifolds.\cite{boissonnat2018delaunay,boissonnat2018obstruction,leibon2000delaunay}

The first row of Figure~\ref{fig:3x3} displays the sublevelset persistent homology from butane to hexane, computed from the OPLS-UA MD simulation data.
Though sampling noise has been added, we note the similarity with the analytical OPLS-UA barcodes.
For butane and pentane, the MD simulation samplings are dense enough so that the OPLS-UA barcodes are nearly indistinguishable from the analytical barcodes.
For hexane, the MD simulation has small gaps in the sampling, which is reflected in the noise added to the hexane OPLS-UA barcodes.

The alkane OPLS-AA PEL includes two additional types of degrees of freedom related to the motions of the hydrogen atoms and represents an important counterpoint to the united atom (UA) PELs. These added terms to the intramolecular potential energy function that describe the molecule change both the absolute value of the potential energy and the topology of the potential energy landscape.
The first added degree of freedom is associated with H--C--H and H--C--C vibrations, which account for the majority of the increase in the absolute value of the potential energy within the OPLS-AA model relative to the united atom model.
The second additional degree of freedom is the rotation of the methyl groups, which have a smaller contributions to the increase in potential energy values.
Both of these degrees of freedom modify the topology of the energy landscape, either by altering the relative energies of basins and barriers, or by adding new features to the PEL, which becomes more complex as the alkane chain length is increased.
To assess the effect of the added degrees of freedom upon the configurational PEL associated with the dihedral angles within the OPLS-AA sampled MD data, we performed a sliding average operation over the above two degrees of freedom to produce a PEL as a function of only the dihedral angles, referred to as {\it OPLS-AA averages}. The new sublevelset barcodes for these PELs are shown in the second row of Figure~\ref{fig:3x3}.
We then compare the changes to the sublevelset persistence barcodes for the OPLS-AA averages and OPLS-UA systems.
For butane, the number of bars of each dimension is the same in both butane OPLS-UA and OPLS-AA, but the energy values corresponding to the births and deaths of these bars have been shifted higher as a result of the added potential energy terms caused by the H-atoms in the OPLS-AA description of the intramolecular interactions.
For pentane, the nine 0-dimensional bars in OPLS-UA are also apparent in the OPLS-AA barcode (with energy values shifted higher).
For 1-dimensional homology, there is a correspondence between the semi-infinite bars, but very weak correspondence between the finite 1-dimensional bars.
Despite the averaging performed on the OPLS-AA data (which should smooth out the effects of rapid H-atom vibrational motion), the OPLS-AA PEL is significantly noisier or rougher than the OPLS-UA PEL for pentane, indicating the effect of the added degrees of freedom upon the dihedral PEL.
Note too, that because of the added degrees of freedom within the MD sampling, it is more difficult to sample all configurations of the vibrational and rotational motion at each dihedral angle value, which lends itself to more sampling noise within the PEL.
The computational cost of MD sampling necessary to compare the OPLS-UA and OPLS-AA PEL within higher-chain alkanes was prohibitive, however the data illustrated in Figure~\ref{fig:3x3} are indicative of the rapid increase in topological complexity that occurs when adding degrees of freedom to even simple systems such as the alkanes.

\subsection{Topological Distances Between ELs}
\label{ssec:distance}

In addition to identifying qualitative differences between ELs, say under a change to the chemical environment, it is also important to be able to measure quantitative distances between ELs.
We are using the word \emph{distance} in the mathematical sense: a distance is a function that accepts two ELs as input, and returns a nonnegative real number which measures how close the two ELs are. Small distances correspond to nearby ELs, and large distances correspond to ELs that are far apart.

One class of distances may be defined on the ELs themselves.
For example, if two ELs have the same domains (input spaces), then one can compute the $L^\infty$ (maximum difference in energy values) or $L^1$ (average difference in energy values) between these two ELs.
Another example of distance between ELs would be the earth mover's or Wasserstein distance,\cite{rubner2000earth} which roughly speaking computes how much physical work would be required to transform one EL into the other.
Notions of earth mover's distance between graph representations of ELs have been previously studied.\cite{cazals2015conformational}

We propose a second class of distances on ELs, which are defined on the sublevelset persistent homology barcodes of the ELs.
Two distances in this class are the {\it bottleneck} and {\it Wasserstein distances} between persistence barcodes.\cite{cohen2007stability,cohen2010lipschitz}
Intuitively, these distances are computed by minimizing the costs of matching one persistence barcode with the other.
For any matching between intervals in one barcode with intervals in another barcode, we can compute the {\it cost} of this matching.
For the bottleneck distance, the cost of a matching is defined to be the largest cost between intervals that are matched, whereas the 1-Wasserstein distance is defined to be the {\it sum} of all of the costs between matched intervals.
The value of the bottleneck distance and the 1-Wasserstein distance is then defined to be the minimum cost over all matchings.

An advantage of defining distances between ELs using persistent homology is that ELs of different dimensionality can be compared in this way.
Indeed, any EL, regardless of its dimension, has sublevelset persistence barcodes, and any two such barcodes can be compared regardless of the dimensionalities of the ELs from which they came.

To investigate the robustness of sublevelset persistent homology on energy landscapes, for each alkane, we may quantify the differences between any two of its three different representations: analytic, OPLS-UA, and OPLS-AA averaged. 
We conduct this analysis by mainly focusing on analytical pentane and pentane OPLS-UA, but also touch on hexane.
We do not perform these quantitative comparisons between OPLS-UA and OPLS-AA Averaged due to their large differences in energy.
When comparing analytical pentane with pentane OPLS-UA, our MD simulations allows for sufficient sampling of the PEL and in return, the barcodes between pentane-UA and analytical pentane are essentially the same (see Figures~\ref{fig:alkane-barcodes} and~\ref{fig:3x3}). 
Therefore, it is not surprising that the bottleneck distance between analytic pentane and pentane OPLS-UA persistence barcodes is small (Table~\ref{tbl:distances}). 
In particular, the bottleneck distance is equal to $0.3$ for $0$-dimensional homology and $0.28$ for $1$-dimensional homology.
However, in the case of hexane OPLS-UA, insufficient sampling of its energy landscape leads to a barcode with shorter bars and delayed births (see Figure~\ref{fig:3x3}).
Thus, the bottleneck distance between analytic hexane and hexane OPLS-UA persistence barcodes is larger.

In comparison to the bottleneck distance, we find that the 1-Wasserstein distance between an analytic alkane and its respective OPLS-UA version serves as a finer quantitative measure of the differences in the EL associated with sampling by MD.
For example, the relatively similar values of the 0- and 1-dimensional bottleneck distances, even in the case of hexane, indicates that the bottleneck distance does not distinguish between poor sampling around a {\it single} local maxima as opposed to bad sampling around {\it many} local maxima.
In contrast, the 1-Wasserstein distance exhibits very large differences between the 0- and 1-dimensional homology, particularly in hexane, where the sampling of regions associated with the transitions between different minima is much poorer than in pentane example. 

\begin{table}
\begin{footnotesize}
\label{tbl:distances}
\begin{tabular}{|c||c|c|} \hline \hline 
 & Bottleneck Distance & \\
\hline
\hline
& $0$-dim.\ homology & $1$-dim.\ homology  \\
\hline
Butane & 0.24 & 0.00 \\ \hline
Pentane & 0.30 & 0.28 \\ \hline
Hexane & 5.03 & 3.60 \\ \hline
\hline
& Wasserstein Distance  & \\
\hline
\hline
& $0$-dim.\ homology & $1$-dim.\ homology  \\ \hline 
 Butane   & 0.45             & 0.00            \\ \hline
 Pentane  & 3.11             & 2.35          \\ \hline
 Hexane   & 76.55            & 160.57           \\ \hline \hline
\end{tabular}
\end{footnotesize}
\caption{
The bottleneck and 1-Wasserstein distances between the persistent barcodes derived from the analytic PEL (Equation~\eqref{eq:alkane-pel}) and sampled PEL from the OPLS-UA molecular dynamics simulations.
}
\label{table:wasserstein}
\end{table}

\section{Mathematical Characterization of the Sublevelset Persistent Homology of Alkanes}
\label{sec:characterization}

In Theorem~\ref{thm:ph-k-main} we provide a complete mathematical characterization of the sublevelset persistent homology of the alkanes, in all homological dimensions, and for any number of carbons in the chain.
We obtained this theorem by computing the sublevelset persistent homology barcodes in Figure~\ref{fig:alkane-barcodes}, conjecturing a formula for how the barcodes would look for any number of carbons in the chain, and then rigorously proving this formula using the connection between Morse theory and persistent homology.
Our proof proceeds by understanding the critical points of the alkanes and their indices; from there we can recover the persistent homology.
In applications, this process will often go in reverse --- computed sublevelset persistent homology may be used as a summary of the critical points in the PEL (perhaps even vectorized for use in machine learning tasks.)

\begin{theorem}
\label{thm:fn-morse}
The function $f_n\colon (S^1)^n\to \R$ is a Morse function with $6^n$ critical points.
Furthermore, $f_n$ has $3^n{n \choose k}$ critical points of index $k$.
\end{theorem}

Indeed, the butane energy landscape in Figure~\ref{fig:EL-butane-pentane}A is a Morse function\cite{milnor2016morse} with three minima and three maxima.
The energy function $f_n\colon (S^1)^n\to \R$ is defined by summing butane functions together,
namely $f_n(\phi_1, \ldots, \phi_n) =  f_1(\phi_1) + ... + f_1(\phi_n)$.
For this reason, the restriction of the pentane PEL to any horizontal or vertical slice produces a translated copy of the butane PEL.
Similarly, the restriction of the hexane PEL to any coordinate-aligned line produces a translated copy of the butane PEL, and the restriction of the hexane PEL to any coordinate-aligned plane produces a translated copy of the pentane PEL.
This structure implies that $f_n$ is also a Morse function with $3^n{n \choose k}$ critical points of index $k$.
Indeed, the $n$ choose $k$ term ${n \choose k}:=\frac{n!}{k!\cdot (n-k)!}$ arises because for $(\phi_1,\ldots,\phi_n)\in (S^1)^n$ to be a critical point of index $k$, we must select $k$ of the angles $\phi_i\in S^1$ to be minima of the butane PEL, and the remaining $n-k$ coordinates maxima for butane.
The term $3^n$ arises because for each of the $\phi_i$, we have three maxima (or minima) of butane to choose from.
The correspondence between Morse theory and sublevelset persistent homology will now allow us to describe the persistent homology bars.
Indeed, each $k$-dimensional bar in the persistent homology barcode has a birth time corresponding to the energy value of a critical point of index $k$, and a death time corresponding to the energy value of a critical point of index $k+1$.

\begin{theorem}
\label{thm:ph-k-main}
The sublevelset persistent homology of the analytical alkane potential energy landscape $f_n \colon (S^1)^n\to \R$ has ${n \choose k}+(3^n-1){n-1 \choose k}$ persistent homology bars in dimension $k$.
\end{theorem}

See Table~\ref{table:ph}.
The number of bars that are semi-infinite is ${n \choose k}$, and the remaining $(3^n-1){n-1 \choose k}$ bars all have the same finite length which is equal to the energy differential between the non-global maxima and the non-global minima in the energy function for butane.
Furthermore, in Appendix~\ref{app:ph-alkanes} we give the birth values for these bars.
This provides a complete characterization of the sublevelset persistent homology of the analytical alkane potential energy landscapes.

\begin{table}
\begin{footnotesize}
\begin{tabular}{|c||c|c|c|c|c|c|c|} \hline & $k=0$ & 1 & 2 & 3 & 4 & 5 & \ldots  \\ \hline \hline
 $n=1$   & 3             & 1             & 0             & 0             & 0             & 0             & \ldots            \\ \hline
 2   & 9             & 10            & 1             & 0             & 0             & 0             & \ldots           \\ \hline
 3   & 27            & 55            & 29            & 1             & 0             & 0             & \ldots          \\ \hline
 4   & 81            & 244           & 246           & 84            & 1             & 0             & \ldots          \\ \hline
 5   & 243           & 973           & 1462          & 978           & 247           & 1             & \ldots         \\ \hline
\vdots & \vdots & \vdots & \vdots & \vdots & \vdots & \vdots & \\ \hline
$n$   & $3^n$      & $3^n(n-1)+1$          & ${n \choose 2}+(3^n-1){n-1 \choose 2}$           & \ldots           & \ldots           & \ldots           & \\ \hline
\end{tabular}
\end{footnotesize}
\caption{
We display the number of $k$-dimensional bars in the sublevelset persistent homology of $f_n \colon (S^1)^n\to \R$; compare with Figure~\ref{fig:alkane-barcodes}.
The entry in row $n$ and column $k$ is equal to ${n \choose k}+(3^n-1){n-1 \choose k}$.}
\label{table:ph}
\end{table}

The main idea behind the proof of Theorem~\ref{thm:ph-k-main} is as follows.
The birth and death values in the sublevelset persistent homology of a Morse function correspond precisely to its critical values.
Since $f_n$ has $3^n$ critical points of index 0, we know that there are $3^n={n \choose 0}+(3^n-1){n-1 \choose 0}$ bars in its 0-dimensional sublevelset persistent homology, as desired.
We can now proceed by induction on $k$, which means that we assume that $f_n$ has ${n \choose k-1}+(3^n-1){n-1 \choose k-1}$ bars of dimension $k-1$ (and we must show that $f_n$ has ${n \choose k}+(3^n-1){n-1 \choose k}$ bars of dimension $k$).
The torus $S^n$ has $(k-1)$-dimensional homology of rank ${n \choose k-1}$, producing ${n \choose k-1}$ semi-infinite bars.
This leaves $(3^n-1){n-1 \choose k-1}$ finite bars in the $(k-1)$-dimensional persistent homology.
Recall that $f_n$ has $3^n{n \choose k}$ critical points of index $k$.
Of these, $(3^n-1){n-1 \choose k-1}$ correspond to death times of the finite $(k-1)$-dimensional bars.
The remaining critical points of index $k$ that do not kill $(k-1)$-dimensional bars instead give birth to $k$-dimensional bars.
So as desired, the number of $k$-dimensional persistent homology bars is the difference
\begin{align*}
\textstyle
3^n{n \choose k} - (3^n-1){n-1 \choose k-1}
&=
\textstyle
{n \choose k}  + (3^n-1)\left({n \choose k}-{n-1 \choose k-1}\right) \\
&=
\textstyle
{n \choose k}  + (3^n-1){n - 1 \choose k}.
\end{align*}
This proof by induction suffices to count the number of persistent homology bars in each homological dimension.
A more detailed argument, relying on Morse theory\cite{milnor2016morse} and a K\"{u}nneth formula for persistent homology,\cite{GakharPerea_KunnethFormula} is required to show that all finite bars have the same length and to determine the birth times of all bars; see Appendix~\ref{app:ph-alkanes}.

\section{Conclusions}
We propose sublevelset persistent homology as a compact representation of energy landscapes that measures more geometric and topological information than a disconnectivity graph.
In particular, 1-dimensional sublevelset persistent homology encodes the energy barriers between nearby transition paths (in analogy with how disconnectivity graphs encode the energy barriers between nearby local minima).
We have chosen the alkanes as our case study model, and have derived a complete formula for their sublevelset persistent homology barcodes for any number of carbons in the chain.
Indeed, the alkanes with $n + 3$ carbons have ${n \choose k} + (3^n - 1){n - 1 \choose k}$ persistent homology bars in homological dimension $k$.
When changing from this analytical model to the OPLS-UA MD simulation data, sampling noise is added to the energy landscape.
As nearby energy landscapes have nearby persistent homology barcodes, we use persistent homology to quantify a notion of topological distance between the two PELs.
A mathematical understanding of the topology of the alkane PEL provides a useful new tool to relate physicochemical properties and configurational phase space.

\section*{Dedication}
The authors dedicate this work to Professors Vesta Coufal, JoAnne Peters, Erica Flapan, and Weibing Gu, whose exceptional mentorship, creativity, and interdisciplinary research have helped bridge the domains of chemistry, physics, computer science and topology to expand our understanding of chemical structure and reactivity.

\section*{Data Availability}
The data that support the findings of this study are available from the corresponding authors upon reasonable request.

\begin{acknowledgments}
  \vspace*{-0.15in}
This material is based upon work supported by the National Science Foundation under Grant No.\ 1934725.
We would also like to acknowledge XSEDE,\cite{towns2014xsede} SiGap hosting services,\cite{pierce2018supporting} and the Apache Airavata gateway middleware framework.\cite{pierce2014apache}
\end{acknowledgments}


\vspace*{-0.2in}
\section*{References}
\vspace*{-0.2in}

\providecommand{\noopsort}[1]{}\providecommand{\singleletter}[1]{#1}%

\appendix

\vspace*{-0.05in}

\section{Proofs of Theorems~\ref{thm:fn-morse} and~\ref{thm:ph-k-main}}\label{app:num-bars}

In this appendix we derive and prove Theorems~\ref{thm:fn-morse} and~\ref{thm:ph-k-main}, which show that the alkane PEL $f_n \colon (S^1)^n\to \R$ is a Morse function with known critical points, and which count the number of bars in the sublevelset persistent homology barcodes. 

\vspace*{-0.05in}
\subsection{Critical Points of the Alkane PEL}
\vspace*{-0.05in}

A smooth function $f\colon M\to \R$ from a manifold $M$ to the real numbers is \emph{Morse} if it has no degenerate critical points, i.e.\ if the Hessian matrix at each critical point is nonsingular.\cite{milnor2016morse}
The \emph{index} of a critical point, roughly speaking, is the number of linearly independent directions in which one can move and have the value of $g$ decrease; more formally the index is the number of negative eigenvalues in the Hessian at this critical point.

From its definition in Section~\ref{ssec:analytical}, or from its image in Figure~\ref{fig:EL-butane-pentane}A, one can see that the butane energy landscape function $V\colon S^1\to \R$ is a Morse function with six critical points: three local minima and three local maxima.
The alkane energy landscape function $f_n \colon (S^1)^n\to \R$, defined by $f_n(\phi_1,\dots , \phi_n) =  V(\phi_1) + \cdots + V(\phi_n)$, decomposes in a way that allows us to identify all of its critical points, and to furthermore describe their indices.

\begin{proof}[Proof of Theorem~\ref{thm:fn-morse}]
We must show that $f_n\colon (S^1)^n\to \R$ is a Morse function with $3^n{n \choose k}$ critical points of index $k$.
Taking partial derivatives, we see that $\frac{\partial f_n}{\partial \phi_i} = \frac{\partial V(\phi_i)}{\partial \phi_i}$.
Therefore the gradient of $f_n$ is the zero vector if and only if each $\phi_i$ is a critical point of $V$. Since $V$ has 6 critical points, it immediately follows that $f_n$ has $6^n$ critical points.

Let $(\phi_1,\dots,\phi_n)$ be a critical point of $f_n$.
Then, the Hessian of $f_n$ at this point is a diagonal matrix with one positive entry on the diagonal for each $\phi_i$ that is a local minimum of $V$, and with one negative entry on the diagonal for each $\phi_i$ that is a local maximum of $V$.
Since all critical points are non-degenerate, $f_n$ is a Morse function.
Furthermore, the index of a critical point $(\phi_1,\dots,\phi_n)$ is the number of $\phi_i$ points that are local maxima of $V$.
There are ${n \choose k}$ ways to choose the $k$ indices for maxima out of the $n$ slots, and for each $\phi_i$ that is a maximum (resp.\ minimum) we have $3$ possible maxima (resp.\ minima) of $V$ to choose from.
Hence the number of critical points of index $k$ is $3^n{n \choose k}$.
\end{proof}

\vspace*{-0.1in}
\subsection{Connection between Morse Theory and Sublevelset Persistence}

The following is an important and well-known lemma in persistent homology that we will reprove in Appendix~\ref{app:morse-complex}.

\begin{lemma}\label{lem_MorsePersistence}
If $f \colon M \to \R$ is a Morse function, then the birth and non-infinite death times in the sublevelset persistent homology correspond precisely to the critical points of $f$.
Each $k$-dimensional bar has birth time corresponding to a critical point of index $k$, and death time either equal to infinity or otherwise corresponding to a critical point of index $k+1$.
Furthermore, the number of semi-infinite bars in dimension $k$ is given by the $k$-dimensional homology of $M$.
\end{lemma}

This lemma shows that the sublevelset persistent homology of a Morse function is closely tied to the critical points of that function.
But in addition to counting the number of critical points of each index, sublevelset persistent homology also shows how the critical points are related to each other (often paired up).

In the case of the alkane PEL $f_n\colon (S^1)^n\to \R$, the number of semi-infinite bars is easy to derive. 
Indeed, once the energy barrier $r$ is larger than the energy value of the global maximum of $f_n \colon (S^1)^n\to \R$, the sublevelsets $f_n^{-1}(-\infty,r]$ are all equal to the $n$-dimensional torus $(S^1)^n$.
The homology groups of the $n$-torus are well understood: the $k$-th homology group $H_k((S^1)^n)$ of the $n$-torus has rank equal to the binomial coefficient ${n \choose k}$.
So the number of semi-infinite bars in the $k$-dimensional persistent homology of $f_n\colon (S^1)^n \to\R$ is ${n \choose k}$.
For example, the number of semi-infinite bars in the persistent homology of pentane ($n=2$), as $k$ increases from 0 to 2, are ${2 \choose 0}=1$, ${2 \choose 1}=2$, ${2 \choose 2}=1$.
The number of semi-infinite bars in the persistent homology of hexane ($n=3$), as $k$ increases from 0 to 3, are ${3 \choose 0}=1$, ${3 \choose 1}=3$, ${3 \choose 2}=3$, ${3 \choose 3}=1$.
For heptane this list is 1, 4, 6, 4, 1, and for octane this list is 1, 5, 10, 10, 5, 1; these numbers are given by the $n$-th row of Pascal's triangle.

\subsection{Proofs of the Number of Bars}\label{sec:proofs}

We can now use the connections between Morse theory and persistent homology to count the number of persistent homology bars, as shown in Table~\ref{table:ph}.
The total number of bars, in all homological dimensions, is given by the following theorem. 

\begin{theorem}\label{thm:ph-total}
The sublevelset persistent homology of the alkane energy landscape $f_n \colon (S^1)^n\to \R$ has $(6^n+2^n)/2$ bars in all homological dimensions.
\end{theorem}

\begin{proof}[Proof of Theorem~\ref{thm:ph-total}]
The torus $(S^1)^n$ has homology of rank ${n \choose k}$ in dimension $k$, and therefore total homology (the sum of the ranks of the homology groups in all dimensions) of rank $\sum_{k=0}^n {n \choose k} = 2^n$.
So, there are $2^n$ semi-infinite bars that start at a critical point of $f_n$ and never die.
These $2^n$ semi-infinite bars ``use up" $2^n$ of the $6^n$ critical points, and the remaining critical points are are paired up to give $(6^n - 2^n)/2$ finite-length bars in the sublevelset persistent homology barcode.
Hence, by Lemma~\ref{lem_MorsePersistence}, the total number of bars in the persistent homology barcode of $f_n\colon (S^1)^n \to \R$ is
\begin{align*}
\Big(2^n\text{ semi-infinite bars}\Big) + \left(\frac{6^n - 2^n}{2}\text{ finite-length bars}\right)\\
= \left(\frac{6^n+2^n}{2}\text{ bars}\right). 
\end{align*}
\vspace*{-0.1in}
\end{proof}

\vspace*{-0.05in}
We are ready to prove our main theorem.

\begin{proof}[Proof of Theorem~\ref{thm:ph-k-main}]
We must show that the number of $k$-dimensional bars in the sublevelset persistent homology of 
$f_n \colon (S^1)^n\to \R$ is ${n \choose k}+(3^n-1){n-1 \choose k}$.
We proceed by induction on $k$.

Our base case is $k=0$.
Since $f_n$ has $3^n$ critical points of index 0, we know that there are $3^n=(3^n-1){n-1 \choose 0}+{n \choose 0}$ bars in its 0-dimensional sublevelset persistent homology.

For the inductive step, assume that $f_n$ has ${n \choose k}+(3^n-1){n-1 \choose k}$ bars in  its $k$-dimensional sublevelset persistent homology.
The homology of the torus is known to have $k$-dimensional homology of rank ${n \choose k}$, and therefore there are ${n \choose k}$ semi-infinite bars.
This leaves $(3^n-1){n-1 \choose k}$ finite bars in the $k$-dimensional sublevelset persistent homology of $f_n$.
By Theorem~\ref{thm:fn-morse}, $f_n$ has $3^n{n \choose k+1}$ critical points of index $k+1$.
Of these, $(3^n-1){n-1 \choose k}$ correspond to death times of the finite $k$-dimensional bars.
By Lemma~\ref{lem_MorsePersistence}, the remaining critical points of index $k+1$ that do not kill $k$-dimensional bars instead give birth to each of the $(k+1)$-dimensional bars.
So the number of $(k+1)$-dimensional persistent homology bars is
\begin{align*}
&3^n{n \choose k+1} - (3^n-1){n-1 \choose k}\\
=& {n \choose k+1}  + (3^n-1)\left({n \choose k+1}-{n-1 \choose k}\right) \\
=&{n \choose k+1}  + (3^n-1){n-1 \choose k+1},
\end{align*}
as desired.
We are done by induction.
\end{proof}

We remark that Theorem~\ref{thm:ph-k-main} implies Theorem~\ref{thm:ph-total}, since
\begin{align*}
&\sum_{k=0}^n \left({n \choose k}+(3^n-1){n-1 \choose k}\right)\\
=& \sum_{k=0}^n {n \choose k}+(3^n-1)\sum_{k=1}^{n-1}{n-1 \choose k} \\
=& 2^n + (3^n-1)2^{n-1} \\
=& 3^n2^{n-1}+2^{n-1} \\
=& (6^n+2^n)/2.
\end{align*}
To see the first equality above we used the fact that ${n-1 \choose n}=0$, and for the second equality we used that $\sum_{k=0}^n{n \choose k}=2^n$.
Nevertheless, we have proven the simpler Theorem~\ref{thm:ph-total} first, as it is a more straightforward demonstration of the tools needed to prove Theorem~\ref{thm:ph-k-main}.

\section{The Morse complex}
\label{app:morse-complex}

We give a mathematical review of the Morse complex, a fundamental object in Morse theory.
Though the Morse complex was not needed to give a count of the alkane persistent homology bars in Theorem~\ref{thm:ph-k-main}, it will be needed to give a complete description of the sublevelset persistent homology (including birth and death levels) in Appendix~\ref{app:ph-alkanes}.

See Section~8.6 of Ref.~\citenum{jost2008riemannian} and Chapter~3 of Ref.~\citenum{BanyagaHurtubise}, for example, for further details on the material on Morse complexes in this section.
We restrict attention to $\Z/2\Z$ coefficients mainly for convenience, so that we do not have to worry about orientations, but also so that our results apply to non-orientable manifolds.
Let $f$ be a Morse function on a Riemannian manifold $M$ which satisfies the
\emph{transversality condition} (also known as the Morse--Smale condition), i.e.\ the stable and unstable manifolds $W^s(p)$ and $W^u(q)$ intersect transversally for all critical points $p$ and $q$.
The stable manifold of a critical point $p$ is the set of all points $x \in M$ such that the limit $t \to +\infty$ of the solution to $\dot{x}(t) = -\nabla f$ starting at $x$ is $p$.
The unstable manifold is the same, but with the limit $t \to -\infty$.
Write $\cc{C}(f)$ for the vector space over $\Z/2\Z$ generated by the set of critical points of $f$.
We denote the index of a critical point $p$ by $\mu(p)$, and we define $\mu(p,q) := \mu(p) - \mu(q)$.
The moduli space of gradient flow paths from $p$ to $q$ will be denoted $\M{f}{p}{q}$, that is, $\M{f}{p}{q}$ is the set of paths $x \colon \R \to M$ which satisfy $\dot{x} = -\nabla f$, $\lim_{t \to \infty} x(t) = q$, and $\lim_{t \to -\infty} x(t) = p$, where we identify $x$ and $\tilde{x}$ if for some $t$ and $\tilde{t}$ we have $x(t) = \tilde{x}(\tilde{t})$.
If $\mu(p,q) = 1$, then $\M{f}{p}{q}$ is a finite set.
There is a grading on $\cc{C}(f)$ by critical point index which is nonnegative and bounded: the set of critical points of index $k$ generate $C_k(f)$.
Define a homomorphism $\bdry\colon \cc{C}(f) \to \cc{C}(f)$ by
\[ \bdry p = \sum_{q \colon \mu(p,q) = 1} (\#\M{f}{p}{q}) q, \]
where $\#\M{f}{p}{q}$ is the number of the flows from $p$ to $q$, modulo two (since we are using $\Z/2\Z$ coefficients).
One can verify that $(\cc{C}(f),\bdry)$ is a chain complex.
The homology of $\cc{C}(f)$ is {\it Morse homology}.
While {\it a priori} Morse homology depends on $f$ and on the Riemannian metric on $M$ (as the gradient depends on such), it can be shown to be isomorphic to singular homology for any such choices.

\subsection{The Morse Complex and Persistence}
\label{app:morse-complex-persistence}

The persistent homology of the sublevelsets of $f$ is closely related to the Morse complex $\cc{C}(f)$.
Let $M^a = f^{-1}(-\infty,a]$ be the sublevelset of $M$ at value $a$.
Then the second Morse lemma guarantees that $M^a$ has the homotopy type of a cell complex with a cell of dimension $k$ for each critical point of index $k$ in $M^a$.
Let $\cc{C}^a(f)$ be the sub-complex of $\cc{C}(f)$ corresponding to the sublevelset $M^a$.
The sublevelset persistent homology of $f$ is the persistent homology of $\cc{C}^a(f)$ filtered by increasing $a$.

We are now prepared to prove Lemma~\ref{lem_MorsePersistence}, which states the following.
If $f \colon M \to \R$ is a Morse function, then the birth and non-infinite death times in the sublevelset persistent homology correspond precisely to the critical points of $f$.
Each $k$-dimensional bar has birth time corresponding to a critical point of index $k$, and death time either equal to infinity or otherwise corresponding to a critical point of index $k+1$.
Furthermore, the number of semi-infinite bars in dimension $k$ is given by the $k$-dimensional homology of $M$.

\begin{proof}[Proof of Lemma~\ref{lem_MorsePersistence}]
By the first Morse lemma, if $[a,b]$ contains no critical points of $f$, then the inclusion $M^a \hookrightarrow M^b$ is a homotopy equivalence (and $\cc{C}^a(f) \iso \cc{C}^b(f)$), so persistence can only change when passing critical points.
Suppose that $a_0 < a_1 < \cdots < a_n$ is a sequence of values interleaving between the critical values of $f$, so that $M^{a_i}$ contains exactly $i$ critical points and $M^{a_n} = M$.
Let $p_i$ be the critical point with value between $a_{i-1}$ and $a_i$, and let $\cc{C}^i(f) = \cc{C}^{a_i}(f)$.
If $p_i$ is of index $k$, then $\cc{C}^i(f)$ and $\cc{C}^{i-1}(f)$ differ only in the $k$-th degree:
\[ C_k^i(f) = C_k^{i-1}(f) \oplus \spnn(p_i) .\]
Here $\spnn(p_i)$ is isomorphic to the field $\Z/2\Z$, with a single generator $p_i$.
We are interested in the corresponding change in homology.
Let $Z_k^i$ and $B_k^i$ denote the cycles and boundaries, respectively.
Note that $p_i$ cannot be in $B_k^i$ because any bounding $(k+1)$-chain would have to have been in $C_{k+1}^{i-1}(f)$.
Thus the only change in homology comes from a change in $Z_k^i$ or $B_{k-1}^{i-1}$.
There are two possibilities:
\begin{enumerate}
\item If $B_{k-1}^i \iso B_{k-1}^{i-1}$, then either $p_i \in Z_k^i$ or $p_i + \sigma \in Z_k^i$ for some $k$-chain $\sigma\in C_k^{i-1}$, and so the rank of $Z_k^i$ (and thus $H_k^i$) increases by one.
\item If not, then $B_{k-1}^i \iso B_{k-1}^{i-1} \oplus \bdry p_i$, and the $(k-1)$-dimensional homology decreases in rank by one.
\end{enumerate}

The semi-infinite bars are the homology of $\cc{C}^{a_n}(f)$, which is exactly the homology of $M$ by the isomorphism between Morse and singular homology.
The assumption that $M^{a_i}$ contains exactly $i$ critical points is true in the generic case, so any Morse function can be perturbed by a small amount to ensure it holds.
If multiple critical points appear at the same critical value, as they do for the alkane energy function, then we impose such a perturbation, possibly losing the ability to canonically identify which critical point corresponds to which birth or death.
\end{proof}

The following two corollaries are consequences of items 1 and 2, respectively, in the proof of the above lemma.

\begin{corollary}\label{cor_zeroBoundary}
If $\bdry p = 0$ for some critical point $p$ in the Morse complex of $f$, then $p$ corresponds to the birth of a bar in persistent homology.
\end{corollary}

\begin{corollary}
If the bar corresponding to critical point $q$ dies at the appearance of critical point $p$, then $q \in \bdry p$.
\end{corollary}

\section{A K\"{u}nneth Formula for the Alkane Persistent Homology}
\label{app:ph-alkanes}

We use the Morse complex to describe the sublevelset persistent homology for butane and pentane.
We then apply a version of the K\"{u}nneth formula for persistent homology\cite{GakharPerea_KunnethFormula} to give a complete description of the sublevelset persistent homology of the alkane molecules.

In particular, we show that all non-infinite bars in the persistent barcode have exactly the same length.
This is visually evident by the persistence barcodes in Figure~\ref{fig:alkane-barcodes}, or equivalently by the persistence diagrams in Figure~\ref{fig:alkane-diagrams}.
A persistence diagram displays the same information as a persistence barcode, just in a different format --- each interval in the barcode is plotted as a point in the plane, with its horizontal and vertical coordinates the birth and death value, resp., of the interval.

\begin{figure*}
\centering
\includegraphics[width=1\textwidth]{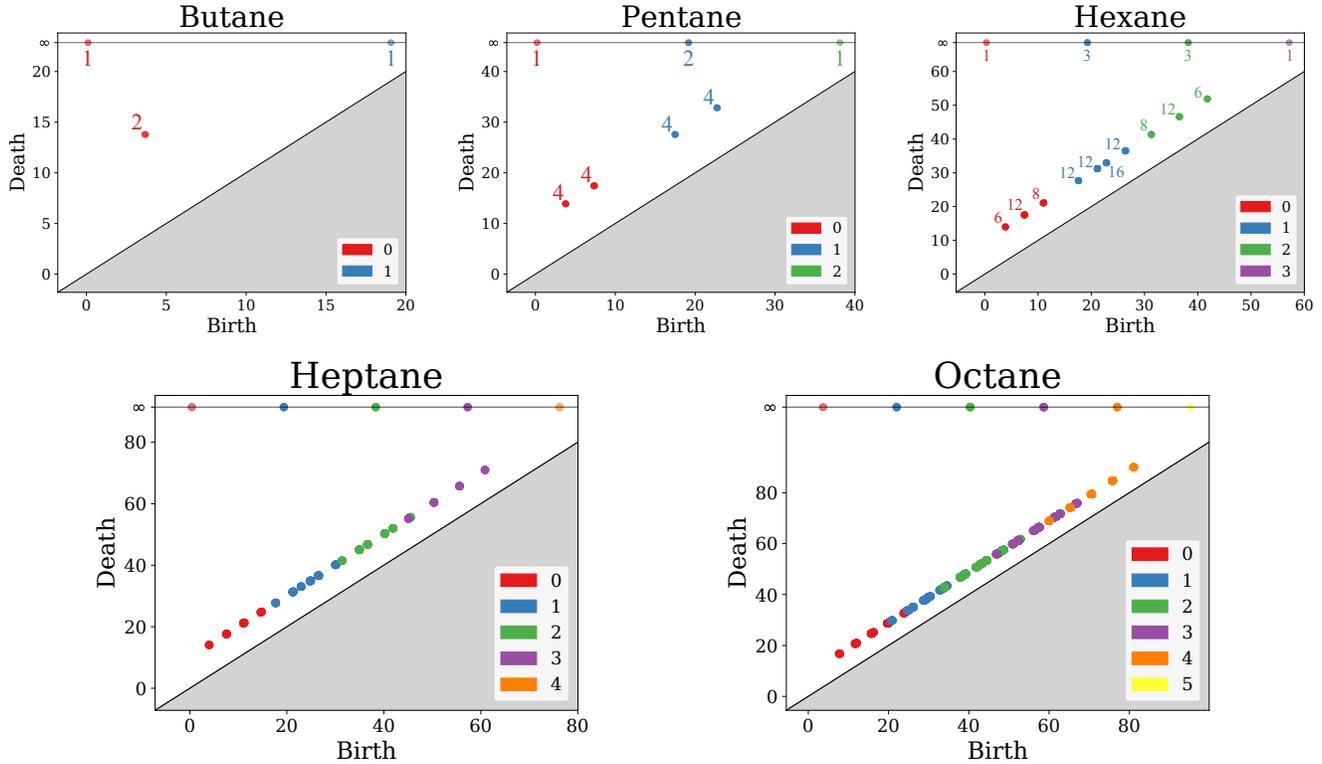}
\caption{
Persistent homology diagrams for the alkane systems, from butane ($n=1$) through octane ($n=5$).
The $x$- and $y$-axes are both energy (kJ/mol).
The color of a persistence diagram point indicates its homological dimension, and the integer label indicates its multiplicity.
The multiplicities for heptane and octane are omitted due to space constraints, but they are given by Theorem~\ref{thm:alkane-ph-characterization}.
}
\label{fig:alkane-diagrams}
\end{figure*}

\subsection{The Persistence of the Butane Energy Function}

The butane energy function $f_1=V\colon S^1\to\R$ has six critical points: a global minimum, $a$, two local minima, $b_1$ and $b_2$, two local maxima, $c_1$ and $c_2$, and a global maximum $d$, with corresponding critical values $\alpha < \beta < \gamma < \delta$; see Figure~\ref{fig:butne-labeled}.

\begin{figure}[ht]
\includegraphics[width=0.4\textwidth]{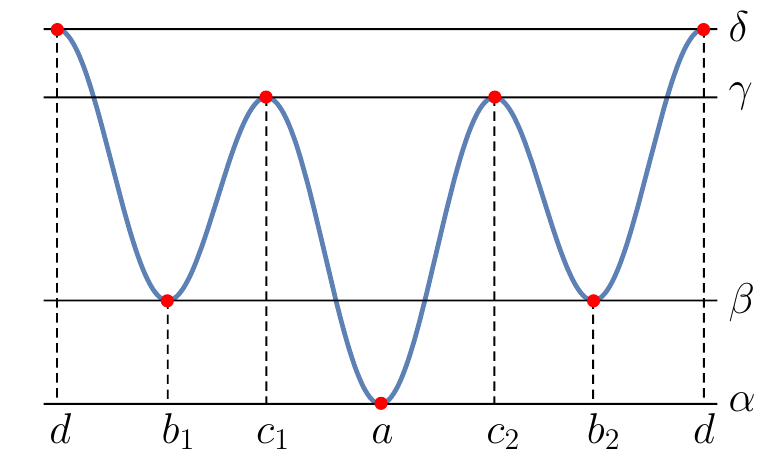}
\caption{PEL of butane.
The $y$-axis is energy, and the $x$-axis is the dihedral angle varying from $0$ to $2\pi$.
From left to right, the critical points are $d, b_1, c_1, a, c_2, b_2$.}
\label{fig:butne-labeled}
\end{figure}

The butane Morse complex is
$ 0 \to (\Z/2\Z)^3 \to (\Z/2\Z)^3 \to 0 .$
We determine the boundary map, $\bdry$, by determining the stable and unstable manifolds of each critical point; see Table~\ref{table:stable-unstable}.
The given intervals are in the circle, $S^1=\R/2\pi\Z$.
The boundary of every index $0$ critical point is $0$.
The boundaries of the index $1$ critical points are
$\bdry c_1 = a + b_1$,
$\bdry c_2 = a + b_2$, and
$\bdry d = b_1 + b_2$.
If necessary for the purpose of breaking ties, assume a small perturbation in the order suggested by the notation: $b_1 \prec b_2$ and $c_1 \prec c_2$.

\begin{table}
\begin{center}
\begin{tabular}{|c||c|c|c|} \hline
$p = $ & index & $W^s(p)$ & $W^u(p)$ \\ \hline\hline
$a$ & 0 & $(c_1,c_2)$ & $\{a\}$ \\ \hline
$b_1$ & 0 & $(d,c_1)$ & $\{b_1\}$ \\ \hline
$b_2$ & 0 & $(c_2,d)$ & $\{b_2\}$ \\ \hline
$c_1$ & 1 & $\{c_1\}$ & $(b_1,a)$ \\ \hline
$c_2$ & 1 & $\{c_2\}$ & $(a,b_2)$ \\ \hline
$d$ & 1 & $\{d\}$ & $(b_2,b_1)$ \\ \hline
\end{tabular}
\caption{The stable and unstable manifolds of the critical points of $f_1=V\colon S^1\to\R$.}
\label{table:stable-unstable}
\end{center}
\end{table}

We will use the elder rule when determining the sublevelset persistence of butane.

\begin{lemma}[Elder Rule\cite{EdelsbrunnerHarer}]\label{lem_elderRule}
If $u$ and $v$ represent distinct $k$-dimensional homology classes in $H_k^i$ at scale $i$, and if $u+v$ is added to the boundaries $B_k^j$ at some scale $j > i$, then the persistent homology bar corresponding to the younger of $u$ and $v$ is killed.
\end{lemma}

We can now state exactly the persistence of the butane energy landscape $f_1=V$.
One zero-dimensional bar is born at $\alpha$, and two are born at $\beta$, by Corollary~\ref{cor_zeroBoundary} above.
At energy level $\gamma$, the appearance of $c_1$ and $c_2$ kills two of the existing bars.
Since the boundaries are both of the form $a + b_i$, the elder rule says that the bars generated by $b_1$ and $b_2$ are those that die.
Lastly a one-dimensional bar appears at $\delta$ since $\bdry d$ is already an element of $B_0$, the 0-dimensional boundary group.
The bars generated by $a$ and $d$ are semi-infinite, matching the homology of $S^1$.
The two finite bars have length $L := \gamma - \beta$.

\subsection{Persistence of the Pentane Energy Function}

\begin{table}
\begin{center}
\begin{tabular}{|c||c|c|c|c|}
\hline
$p = $    & index & $\bdry p$ & $f_2(p)$ & effect     \\ \hline\hline
$(a,a)$     & $0$ & $0$       & $2\alpha$ & birth         \\ \hline
$(a,b_1)$   & $0$ & $0$       & $\alpha+\beta$ & birth    \\ \hline
$(a,b_2)$   & $0$ & $0$       & $\alpha+\beta$ & birth    \\ \hline
$(b_1,a)$   & $0$ & $0$       & $\alpha+\beta$ & birth    \\ \hline 
$(b_2,a)$   & $0$ & $0$       & $\alpha+\beta$ & birth    \\ \hline
$(b_1,b_1)$ & $0$ & $0$       & $2\beta$       & birth    \\ \hline
$(b_1,b_2)$ & $0$ & $0$       & $2\beta$       & birth    \\ \hline
$(b_2,b_1)$ & $0$ & $0$       & $2\beta$       & birth    \\ \hline
$(b_2,b_2)$ & $0$ & $0$       & $2\beta$       & birth    \\ \hline
$(a,c_1)$   & $1$ & $(a,a) + (a,b_1)$      & $\alpha+\gamma$ & death $\alpha+\beta$ bar
\\ \hline
$(a,c_2)$   & $1$ & $(a,b_2) + (a,a)$      & $\alpha+\gamma$ & death $\alpha+\beta$ bar
\\ \hline
$(c_1,a)$   & $1$ & $(a,a) + (b_1,a)$       & $\alpha+\gamma$ & death $\alpha+\beta$ bar
\\ \hline
$(c_2,a)$   & $1$ & $(b_2,a) + (a,a)$       & $\alpha+\gamma$ & death $\alpha+\beta$ bar
\\ \hline
$(b_1,c_1)$ & $1$ & $(b_1,a) + (b_1,b_1)$  & $\beta+\gamma$  & birth  \\ \hline 
$(b_1,c_2)$ & $1$ & $(b_1,b_2) + (b_1,a)$  & $\beta+\gamma$  & birth  \\ \hline 
$(b_2,c_1)$ & $1$ & $(b_2,a) + (b_2,b_1)$  & $\beta+\gamma$  & birth  \\ \hline 
$(b_2,c_2)$ & $1$ & $(b_2,b_2) + (b_2,a)$  & $\beta+\gamma$  & birth  \\ \hline
$(c_1,b_1)$ & $1$ & $(a,b_1) + (b_1,b_1)$   & $\beta+\gamma$  & death $2\beta$ bar
\\ \hline
$(c_1,b_2)$ & $1$ & $(a,b_2) + (b_1,b_2)$   & $\beta+\gamma$  & death $2\beta$ bar
\\ \hline
$(c_2,b_1)$ & $1$ & $(b_2,b_1) + (a,b_1)$   & $\beta+\gamma$  & death $2\beta$ bar
\\ \hline
$(c_2,b_2)$ & $1$ & $(b_2,b_2) + (a,b_2)$   & $\beta+\gamma$  & death $2\beta$ bar
\\ \hline
$(a,d)$     & $1$ & $(a,b_1) + (a,b_2)$    & $\alpha+\delta$ & birth \\ \hline
$(d,a)$     & $1$ & $(b_1,a) + (b_2,a)$     & $\alpha+\delta$ & birth \\ \hline
$(b_1,d)$   & $1$ & $(b_1,b_1) + (b_1,b_2)$& $\beta+\delta$  & birth \\ \hline
$(b_2,d)$   & $1$ & $(b_2,b_1) + (b_2,b_2)$& $\beta+\delta$  & birth \\ \hline
$(d,b_1)$   & $1$ & $(b_1,b_1) + (b_2,b_1)$ & $\beta+\delta$  & birth \\ \hline
$(d,b_2)$   & $1$ & $(b_1,b_2) + (b_2,b_2)$ & $\beta+\delta$  & birth \\ \hline
$(c_1,c_1)$ & $2$ & \begin{tabular}{@{}c@{}}$(a,c_1) + (b_1,c_1)$ \\ $+ (c_1,a) + (c_1,b_1)$\end{tabular} & $2\gamma$ & death $\beta+\gamma$ bar \\ \hline
$(c_1,c_2)$ & $2$ & \begin{tabular}{@{}c@{}}$(a,c_2) + (b_1,c_2)$ \\ $+ (c_1,b_2) + (c_1,a)$\end{tabular} & $2\gamma$ & death $\beta+\gamma$ bar \\ \hline
$(c_2,c_1)$ & $2$ & \begin{tabular}{@{}c@{}}$(b_2,c_1) + (a,c_1)$ \\ $+ (c_2,a) + (c_2,b_1)$\end{tabular} & $2\gamma$ & death $\beta+\gamma$ bar \\ \hline
$(c_2,c_2)$ & $2$ & \begin{tabular}{@{}c@{}}$(b_2,c_2) + (a,c_2)$ \\ $+ (c_2,b_2) + (c_2,a)$\end{tabular} & $2\gamma$ & death $\beta+\gamma$ bar \\ \hline
$(c_1,d)$ & $2$ & \begin{tabular}{@{}c@{}}$(a,d) + (b_1,d)$ \\ $+ (c_1,b_1) + (c_1,b_2)$\end{tabular} & $\gamma+\delta$ & death $\beta+\delta$ bar \\ \hline
$(c_2,d)$ & $2$ & \begin{tabular}{@{}c@{}}$b_2,d) + (a,d)$ \\ $+ (c_2,b_1) + (c_2,b_2)$\end{tabular} & $\gamma+\delta$ & death $\beta+\delta$ bar \\ \hline
$(d,c_1)$ & $2$ & \begin{tabular}{@{}c@{}}$(b_1,c_1) + (b_2,c_1)$ \\ $+ (d,a) + (d,b_1)$\end{tabular} & $\gamma+\delta$ & death $\beta+\delta$ bar \\ \hline
$(d,c_2)$ & $2$ & \begin{tabular}{@{}c@{}}$(b_1,c_2) + (b_2,c_2)$ \\ $+ (d,b_2) + (d,a)$\end{tabular} & $\gamma+\delta$ & death $\beta+\delta$ bar \\ \hline
$(d,d)$ & $2$ & \begin{tabular}{@{}c@{}}$(b_1,d) + (b_2,d)$ \\ $+ (d,b_1) + (d,b_2)$\end{tabular} & $2\delta$ & birth  \\ \hline
\end{tabular}
\end{center}
\caption{Pentane Morse complex computation.
By ``death $r$ bar" in the effect column, we mean that the critical point kills a persistent homology bar that was born at energy $r$.}
\label{table:pentane}
\end{table}

We next consider the pentane energy landscape $f_2\colon (S^1)^2\to\R$.
We know the cells of the Morse complex by our work in Section~\ref{sec:proofs}.
We can now write down the boundary maps and compute the persistence by brute force using Lemma~\ref{lem_MorsePersistence}.
The birth times occur in the order
\begin{align*}
&2\alpha < \alpha+\beta < 2\beta < \alpha+\gamma < \beta+\gamma \\
<& \alpha+\delta < \beta+\delta < 2\gamma < \gamma+\delta < 2\delta.
\end{align*}
Table~\ref{table:pentane} summarizes the persistent homology computation.

We remark that the index $1$ critical points appearing at $\beta+\gamma$ are paired so that half of them are births and half are deaths, but there is no canonical choice of which critical point in each pair causes each effect.

\subsection{General Form}

The direct computations used above for pentane become too complicated to perform by hand for larger alkanes with higher-dimensional PELs.
To proceed to the general case we need several lemmas.
The alkane energy function for $n$ dimensions is $f_n(\phi_1,\ldots,\phi_n) = V(\phi_1) + \ldots V(\phi_n)$.
In Theorem~\ref{thm:fn-morse} we established that the critical points of $f_n$ must be critical points of $V$ in each component, and that the index is the sum of the indices of the components.

We state the following lemmas for the slightly more general setting where manifold $M$ is a product $M = M_1 \times \cdots \times M_n$, Morse function $G \colon M \to \R$ is $G(x_1,\ldots,x_n) = g_1(x_1) + \cdots + g_n(x_n)$, and the functions $g_i \colon M_i \to \R$ are possibly different with possibly different domains.

\begin{lemma}\label{lem_productManifolds}
When product manifold $M$ and Morse function $G(x_1,\ldots,x_n) = g_1(x_1) + \cdots + g_n(x_n)$ are as described above,
\begin{enumerate}
\item A point $p \in M$ is a critical point of $G$ if and only if $p = (p_1, \ldots , p_n)$ with $p_i$ a critical point of $g_i$ for all $i$,
\item The stable manifold of $p$ is $W^s(p) = W^s(p_1) \times \cdots \times W^s(p_n)$ and the unstable manifold is $W^u(p) = W^u(p_1) \times \cdots \times W^u(p_n)$,
\item If $g_1, \ldots , g_n$ each satisfy the transversality condition, so does $G$,
\item The moduli space of gradient flows is $\M{G}{p}{q} = \M{g_1}{p_1}{q_1} \times \cdots \times \M{g_n}{p_n}{q_n}$, 
\item The index of $p$ is $\sum_i \mu(p_i)$ and the relative index of critical points $p$ and $q$ is given by $\mu(p,q) = \sum_{i=1}^n \mu(p_i,q_i)$, and
\item If $\mu(p,q) = 1$ and $\M{G}{p}{q} \neq \emptyset$, then for some $1 \leq i \leq n$ we have $p_i \neq q_i$, $q_j = p_j$ for all $j \neq i$, and
$\#\M{G}{p}{q} = \#\M{g_i}{p_i}{q_i}$.
\end{enumerate}
\end{lemma}

\begin{proof}
\begin{enumerate}
\item Follows by the linearity of the gradient.
\item If $\lim_{t \to \infty} x(t) = p$ where $\dot{x}(t) = -\nabla G(x(t))$, then $\lim_{t \to \infty} x_i(t) = p_i$ for each $i$, and likewise for the limit $t \to -\infty$.
\item To see that $G$ satisfies the transversality condition, recall that $T_xM = T_{x_1}M_1 \times \cdots \times T_{x_n}M_n$, and that by transversality in each component, $T_{x_i}M_i = T_{x_i}W^u(p_i) \oplus T_{x_i}W^s(p_i)$.
Thus $T_xM \iso T_xW^u(p) \oplus T_xW^s(p)$ with the isomorphism given by permuting the order of the coordinates.
\item Follows because a flow $x(t)$ from $p$ to $q$ must consist of a flow $x_i(t)$ from $p_i$ to $q_i$ in each component.
\item Since $\nabla G$ is the sum of $\nabla g_i$, the Hessian of $G$ breaks down into a block-diagonal matrix (in the natural product coordinates) with blocks given by the Hessians of each individual $g_i$.
Therefore the number of negative (respectively, positive) eigenvalues of the Hessian at $p$ is the sum of the number of negative (positive) eigenvalues of each $g_i$, and so $\mu(p) = \sum_i \mu(p_i)$ and $\mu(p,q) = \sum_i \mu(p_i,q_i)$.
\item Let $\mu(p,q) = 1$ and $\M{G}{p}{q} \neq \emptyset$.
Then it must be the case that $p$ and $q$ differ in exactly one coordinate, since otherwise there would be some index $j$ with $\mu(q_j)>\mu(p_j)$, giving $\M{g_j}{p_j}{q_j} = \emptyset$ and $\M{G}{p}{q} = \emptyset$ by 3, a contradiction.
Hence $p$ and $q$ differ in exactly one coordinate $i$, and in this coordinate $\mu(p_i,q_i)=1$, by 5.
So $\M{G}{p}{q} \cong \M{g_i}{p_i}{q_i}$ and $\#\M{G}{p}{q} = \#\M{g_i}{p_i}{q_i}$.
\end{enumerate}
\end{proof}

\begin{corollary}\label{cor_boundary}
In $C_\bullet(G)$ the differential is given by 
\[ \bdry p = \sum_{i=1}^n (p_1, \ldots , \bdry p_i , \ldots , p_n).\]
\end{corollary}

\begin{proof}
By the definition of the Morse boundary operator, $\bdry p = \sum (\#\M{G}{p}{q})q$, where the sum is over all $q$ such that there is a gradient flow from $p$ to $q$ and $\mu(q) = \mu(p) - 1$.
By the lemma, if $p = (p_1, \ldots , p_n)$, then $\bdry p$ is a linear combination of terms of the form $q = (p_1, \ldots , q_i, \ldots , p_n)$, and where a term $(p_1, \ldots , q_i, \ldots , p_n)$ occurs if and only if there is a flow in $M_i$ from $p_i$ to $q_i$ and $\mu(p_i,q_i) = 1$.
The result follows since $\#\M{G}{p}{q} = \#\M{g_i}{p_i}{q_i}$ for this choice of $p$ and $q$.
\end{proof}

The tensor product of chain complexes $\cc{A}$ and $\cc{B}$ has
\[ (\cc{A} \otimes \cc{B})_k = \bigoplus_{i+j=k} (A_i \otimes B_j) \]
as its $k$-th chain group, and the differential is the linear extension of $\bdry(a,b) = (\bdry_A a,b) + (a,\bdry_B b)$ when using $\Z/2\Z$ coefficients.
For an $n$-fold tensor product, a simple induction argument shows that the direct sum runs over all $i_1 + \cdots + i_n = k$, namely 
\[ \left(\cc{A}^{(1)}\otimes\cdots\otimes\cc{A}^{(n)}\right)_k=\bigoplus_{i_1+\ldots+i_n=k}\left(A^{(1)}_{i_1}\otimes\cdots\otimes A^{(n)}_{i_n}\right).\]
The differential generalizes to the $n$-fold case as
\[ \bdry(a_1, \ldots , a_n) = \sum_{i=1}^n (a_1 , \ldots , \bdry_i a_i , \ldots , a_n), \]

\begin{proposition}\label{prop_tensor}
$\cc{C}(G) \iso \cc{C}(g_1) \otimes \cdots \otimes \cc{C}(g_n)$
\end{proposition}

\begin{proof}
By Lemma~\ref{lem_productManifolds}, critical points of $G$ with index $k$ are in bijection with tuples $(p_1, \ldots , p_n)$ where $p_i$ is a critical point of $g_i$ and $\sum \mu(p_i) = k$.
Since such tuples freely generate $(\cc{C}(g_1) \otimes \cdots \otimes \cc{C}(g_n))_k$ and critical points of index $k$ freely generate $C_k(G)$, there is an isomorphism of graded vector spaces $\cc{C}(G) \iso \cc{C}(g_1) \otimes \cdots \otimes \cc{C}(g_n)$.
By Corollary~\ref{cor_boundary} and the definition of the tensor product differential, this extends to an isomorphism as chain complexes.
\end{proof}

A filtered chain complex is a functor from $(\Z,\le)$ to the category of chain complexes.
The tensor product of {\it filtered} chain complexes $\cc{A}^*$ and $\cc{B}^*$, denoted $(\cc{A} \otimes_f \cc{B})^*$, has
\[ (\cc{A} \otimes_f \cc{B})^c_k := \bigoplus_{a+b \le c}\  \bigoplus_{i+j=k} (A_i^a \otimes B_j^b) \]
as its $k$-th chain group at filtration value $c$, and boundary maps given by the restriction of the usual tensor product.

The Morse complex of a Morse function $f$ is naturally a filtered chain complex $\cc{C}^*(f)$, where $\cc{C}^a(f)$ is the subcomplex of $\cc{C}(f)$ generated by critical points with value less than $a$; this is a functor with domain $(\Z,\le)$ so long as we restrict attention to a discrete set of real numbers $a$ that are interleaved between adjacent critical points of $f$.
The differential on $\cc{C}^a(f)$ is the restriction of the differential of $\cc{C}(f)$ which is well-defined because the gradient flow construction of the boundary ensures that $f(q) < f(p)$ for any $q \in \bdry p$.

\begin{theorem}
\label{thm:chain-iso-product}
Let $M$ be a product of manifolds $M = M_1 \times \cdots \times M_n$, and let Morse function $G \colon M \to \R$ be defined by $G(x_1,\ldots,x_n) = g_1(x_1) + \cdots + g_n(x_n)$, where each function $g_i \colon M_i \to \R$ is Morse and satisfies the transversality condition.
Then $\cc{C}^*(G) \iso \cc{C}^*(g_1) \otimes_f \cdots \otimes_f \cc{C}^*(g_n)$ as filtered chain complexes.
\end{theorem}

\begin{proof}
By Proposition~\ref{prop_tensor}, $\cc{C}(G) \iso \cc{C}(g_1) \otimes_f \cdots \otimes_f \cc{C}(g_n)$ as chain complexes.
Since $G(x_1,\ldots,x_n) = g_1(x_1) + \cdots + g_n(x_n)$, for a critical point $p$, $G(p) < a$ if and only if $g_1(p_1) + \cdots + g_n(p_n) < a$.
Thus the natural filtration of $\cc{C}^*(G)$ agrees with the natural filtration of $\cc{C}^*(g_1) \otimes_f \cdots \otimes_f \cc{C}^*(g_n)$.
\end{proof}

The homology of a tensor product can be computed using the K\"unneth formula (see Theorem~3B.5 of Ref.~\citenum{Hatcher}, for example).

\begin{figure*}[tb]
\centering
\includegraphics[width=1\textwidth]{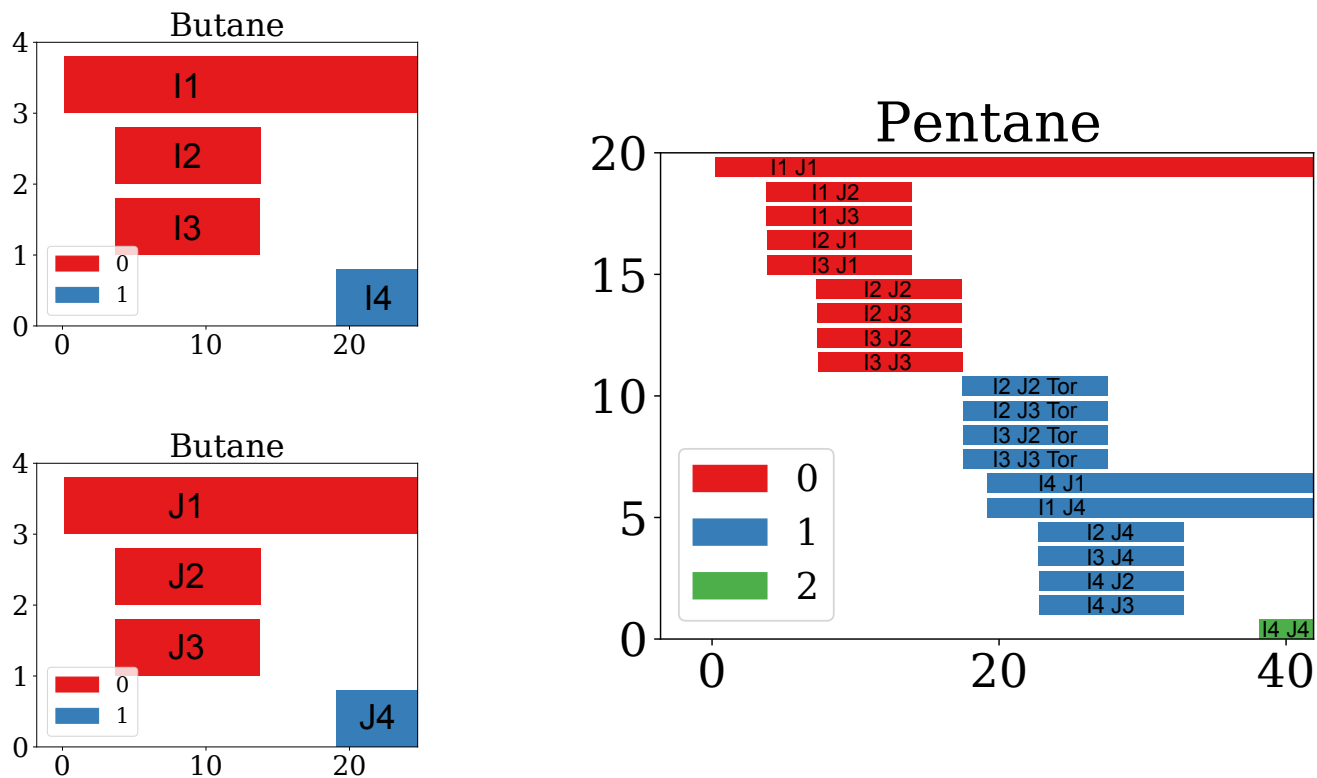}
\caption{We obtain the pentane barcodes by applying the persistent K\"{u}nneth formula to two copies of the butane barcodes.
The bar ``I2 J2" in pentane comes from the butane bars I2 and J2, and the bar ``I2 J2 Tor" is from the torsion portion of Theorem~\ref{thm:persistent-kunneth}.
The $y$-axis is a count of bars, and the $x$-axis is energy (kJ/mol).}
\label{fig:KunnethAlkanes}
\end{figure*}

\begin{theorem}[K\"unneth Formula]
The homology of two chain complexes $\cc{A}$ and $\cc{B}$ satisfies
\[ \bigoplus_{i+j=n} (H_i(\cc{A}) \otimes H_j(\cc{B})) \iso H_n(\cc{A} \otimes \cc{B}) .\]
\end{theorem}

This theorem is insufficient to give the persistent homology, however.
For example, in the case of the alkanes, $f_2^{-1}(-\infty,\beta+\alpha)$ includes the points $(b_1,a)$ and $(a,b_1)$, but it does not include $(b_1,b_1)$, so the chain complex is not the ``pointwise'' tensor product.

Instead, we need the version of the K\"unneth formula for persistent homology by Gakhar and Perea;\cite{GakharPerea_KunnethFormula} see their Theorems~5.12 and 5.14.
The persistent homology of a filtered chain complex $\cc{A}^*$ in degree $k$ is denoted $PH_k(A)$.
We denote the space of barcodes of $\cc{A}^*$ by $\bcd(A)$, where $\bcd_j(A)$ denotes the bars in homological dimension $j$.

\begin{theorem}[Peristent K\"unneth Formula]
\label{thm:persistent-kunneth}
There is a natural short exact sequence of graded modules
\begin{align*}
0 \to \bigoplus_{i+j=n} \left( PH_i(X) \otimes PH_j(Y)\right) \to PH_n(X \otimes_g Y) \\
\to \bigoplus_{i+j=n} \Tor(PH_i(X),PH_{j-1}(Y)) \to 0 .
\end{align*}
If $H_i(X)$ and $H_j(Y)$ are pointwise finite, then
\begin{footnotesize}
\begin{align*}
&\ \bcd_n(X \otimes_f Y) \\
=&\ \ \bigsqcup_{i+j=n} \left\{ (\ell_J + I) \cap (\ell_I + J) ~|~ I \in \bcd_i(X), J\in \bcd_j(Y) \right\}\\
& \sqcup \bigsqcup_{i+j=n} \left\{ (r_J + I) \cap (r_I + J) ~|~ I \in \bcd_i(X), J\in\bcd_{j-1}(Y) \right\} .
\end{align*}
\end{footnotesize}
Here $\ell$ and $r$ are the left and right endpoints of the interval.
\end{theorem}

By convention if the right endpoint of interval $I$ is $r_I = +\infty$, then the bar $r_I + J$ does not appear.

\begin{corollary}\label{cor_infiniteBars}
Critical point $(p , q)$ generates a semi-infinite bar if and only if $p$ and $q$ generate semi-infinite bars.
\end{corollary}

We are now prepared to prove that all non-infinite bars in the sublevelset persistent homology have the same length; see Figures~\ref{fig:alkane-barcodes} and~\ref{fig:alkane-diagrams}.

\begin{corollary}\label{cor_finiteLength}
If $g_1 = \cdots = g_n$ and all finite bars in the persistence of $g_1$ are of the same length, $L$, then all finite bars in the persistence of $G(x_1,\ldots,x_n):=\sum_{i=1}^n g_i(x_i)$ have length $L$.
\end{corollary}

\begin{proof}
By assumption this holds for the case $n=1$.
Suppose it holds for $n-1$.
A finite bar is the intersection of a bar in $PH_i(G_{n-1})$ with one in $PH_j(g_n)$, at least one of which is finite.
The length of an intersection of bars (with starting endpoints shifted to be identical) is equal to the length of the shorter of the two bars, which in this case is length $L$.
\end{proof}

\subsection{Persistence of the Alkanes}\label{sec:pers-description}

We apply Theorem~\ref{thm:persistent-kunneth} and Corollaries~\ref{cor_infiniteBars} and~\ref{cor_finiteLength} to the sublevelset persistence of the alkane energy function, $f_n(\phi_1,\ldots,\phi_n) = \sum_{i=1}^n f_1(\phi_i)$.
See Figure~\ref{fig:KunnethAlkanes}.
These computations are based off of our computations for butane, from which we know that the semi-infinite bars of butane are those generated by critical points $a$ and $d$.

\begin{definition}
Let $f_n\colon (S^1)^n\to\R$ be the alkane energy function, and let $k\le n$.
For $i\le k$ and $j\le n-k$, we say an index $k$ critical point $(\phi_1,\ldots,\phi_n)$ of $f_n$ is of {\it class $(k,i,j)$} if the ordered list of $\phi_\ell$ points consists of
\begin{itemize}
\item $i$ copies of $c_1$ or $c_2$, and hence $k-i$ copies of $d$, and
\item $j$ copies of $b_1$ or $b_2$, and hence $n-k-j$ copies of $a$.
\end{itemize}
\end{definition}

The motivation for this definition is revealed in the following immediate lemma.

\begin{lemma}
All critical points of $f_n\colon (S^1)^n\to\R$ of class $(k,i,j)$ have the same energy value
$E(n,k,i,j):=i\gamma+(k-i)\delta+j\beta+(n-k-j)\alpha$.
\end{lemma}

\begin{proof}
This follows since $f_n(\phi_1,\ldots,\phi_n) = \sum_{i=1}^n f_1(\phi_i)$, where $f_1(c_1)=f_1(c_2)=\gamma$, $f_1(d)=\delta$, $f_1(b_1)=f_1(b_2)=\beta$, and $f_1(a)=\alpha$.
\end{proof}

Let $n_1,n_2,\ldots, n_m$ be integers with $n_1+n_2+\ldots+n_m=n$.
The {\it multinomial coefficient}, which is a generalization of the binomial coefficient, is defined as
\[ \binom{n}{n_1,\ldots,n_m} := \frac{n!}{n_1!\cdot \ldots \cdot n_m!}.\]
It is equal to the number of ways, from a collection of $n$ objects, to choose $n_1$ objects to go in box 1, to choose $n_2$ objects to go in box 2, \ldots, and to choose $n_m$ objects to go in box $m$.
Note that $0!=1$.
If any of the integers $n_i$ are negative or greater than $n$, then $\binom{n}{n_1,\ldots,n_m}=0$.

\begin{lemma}\label{lem:num-class-kij}
The number of critical points of class $(k,i,j)$ is $2^{i+j}\binom{n}{i,k-i,j,n-k-j}$.
\end{lemma}

\begin{proof}
Among its $n$ entries $(\phi_1,\ldots,\phi_n)$, a critical point of class $(k,i,j)$ has $i$ copies of $c_1$ or $c_2$, $k-i$ copies of $d$, $j$ copies of $b_1$ or $b_2$, and $n-k-j$ copies of $a$.
Hence the lemma follows from the definition of the multinomial coefficient, where the constant $2^{i+j}$ appears because there are two choices for each of the $i$ copies of $c_1$ or $c_2$, and there are two choices for each of the $j$ copies of $b_1$ or $b_2$.
\end{proof}

The following theorem gives the complete classification of the sublevelset persistent homology of all of the alkanes.

\begin{theorem}
\label{thm:alkane-ph-characterization}
Consider the $k$-dimensional sublevelset persistent homology barcodes of the alkane PEL $f_n\colon (S^1)^n\to\R$.
Let $i\le k$ and let $j\le n-k$.
The number of bars that appear in the $k$-dimensional sublevelset persistent homology with birth energy value equal to $E(n,k,i,j)=i\gamma+(k-i)\delta+j\beta+(n-k-j)\alpha$ is
\[ 2^{i+j}\left(\sum_{\ell=0}^{i}(-1)^{\ell}\binom{n}{i-\ell,k-i,j+\ell,n-k-j}\right). \]
Furthermore, this bar is semi-infinite if and only if $i=j=0$, and otherwise has length $L=\gamma-\beta$.
These are the only bars that appear.
\end{theorem}

We check that when $i=j=0$, we get $\binom{n}{0,k,0,n-k}=\binom{n}{k,n-k}=\binom{n}{k}$ semi-infinite bars, as expected.

\begin{proof}[Proof of Theorem~\ref{thm:alkane-ph-characterization}]
Fix $n$ to be arbitrary.
We will induct on $k\le n$.

For the base case $k=0$, note that necessarily $i=0$.
The formula then follows from Lemma~\ref{lem:num-class-kij}, since the number of critical points of class $(k,i,j)$ is $2^{i+j}\binom{n}{i,k-i,j,n-k-j}$, and each of those necessarily gives birth to a 0-dimensional peristent homology bar at the corresponding energy value.

For the inductive step, assume that the formula is true for $k-1$, i.e.\ for all $i\le k-1$ and $j\le n-k+1$.
Our task is now to prove the formula is true for $k$, i.e.\ for all $i\le k$ and $j\le n-k$.

By Lemma~\ref{lem:num-class-kij} the number of critical points of class $(k,i,j)$ is $2^{i+j}\binom{n}{i,k-i,j,n-k-j}$.
Of those, the number of critical points that must be used to kill $(k-1)$-dimensional bars is equal, by induction, to
\[ 2^{i+j}\left(\sum_{\ell=0}^{i-1}(-1)^{\ell}\binom{n}{i-1-\ell,k-i,j+1+\ell,n-k-j}\right). \]
Indeed, if a critical point of class $(k,i,j)$ kills a $(k-1)$-dimensional bar of length $L$, then the birth time of that bar must have been the energy of a critical point of class $(k-1,i-1,j+1)$ (Note $E(n,k,i,j)-L=E(n,k-1,i-1,j+1)$).
Hence the number of $k$-dimensional persistent homology bars with birth energy equal to $E(n,k,i,j)$ is
\begin{align*}
&\textstyle{2^{i+j}\binom{n}{i,k-i,j,n-k-j}} \\
&\textstyle{- 2^{i+j}\left(\sum_{\ell=0}^{i-1}(-1)^{\ell}\binom{n}{i-1-\ell,k-i,j+1+\ell,n-k-j}\right)} \\
=&\textstyle{2^{i+j}\left({\scriptstyle\binom{n}{i,k-i,j,n-k-j}-\sum\limits_{\ell=0}^{i-1}(-1)^{\ell}\binom{n}{i-1-\ell,k-i,j+1+\ell,n-k-j}}\right)} \\
=&\textstyle{2^{i+j}\left(\sum_{\ell=0}^{i}(-1)^{\ell}\binom{n}{i-\ell,k-i,j+\ell,n-k-j}\right).}
\end{align*}

\vspace*{-0.2in}
\end{proof}

Thus we have given a complete description of the sublevelset persistent homology of the alkanes, for all dimensions $n$, and for all homological dimensions $k$.

\section{The Morse--Smale Complex of the PEL of alkanes}
\label{app:morse-smale}

\vspace*{-0.1in}
When considering a given conformation we are interested in what conformations are likely to be obtained next as energy in the system increases or decreases.
For a non-critical point in the conformation space $X$, the paths $x:[0,1]\to X$ which pass through this point and are integral curves to the energy gradient begin and end at the critical points.
The energy surface --- whether reduced or not --- can therefore be partitioned into regions wherein gradient flow lines share endpoints.
Two points from the same region correspond to conformations which tend towards the same arrangement as energy in the system changes.
Such a partition can be formally constructed as the Morse--Smale complex for the given energy function $f$, provided that $f$ is a Morse function which additionally satisfies the transversality condition; see Appendix~\ref{app:morse-complex}.
The Morse--Smale complex is a cell complex on $X$ whose $n$-cells --- that is the pieces of the decomposition homeomorphic to $n$-dimensional Euclidean balls --- correspond to pairs of critical points $p$ and $q$ whose indices differ by $n$ and which are the endpoints of integral curves through points in the interior of the cell.
(We remark that the Morse--Smale complex is a different object than the Morse complex in Appendix~\ref{app:morse-complex}, though the two are closely related.)

The potential energy landscape for butane has six critical points, and so the Morse--Smale complex has six $0$-cells.
Three of these are minima and three are maxima, with indices $0$ and $1$ respectively.
Each minimum is connected via $1$-cells in the Morse--Smale complex to its two neighboring maxima, indicating the respective corresponding flows.
As the landscape is one-dimensional, this exhausts the complex.

The conformation space of pentane has two periodic dimensions, so critical points can occur with index $0$, $1$, or $2$.
There are nine critical points of index $0$ and also of index $2$, with the remaining $18$ being index $1$.
At horizontal and vertical slices, the pentane landscape is a copy of the butane landscape.
Hence, the Morse--Smale complex is also a copy at these slices.
In particular, the $1$-cells of the complex for pentane occur along a single varying dihedral angle.
Here though, there additionally are $2$-cells for each pair of adjacent minima and maxima.
These fill the regions between the $1$-cells joining saddles and extrema, and correspond to an uncountable family of flows from a minimum to a maximum.
Such flows never come to a stop at a saddle point, with both dihedral angles varying as the molecule changes directly from one extremum to another.
We used the Topology ToolKit (TTK)\cite{tierney2018ttk} to calculate the Morse--Smale complex and produce a visualization on a uniform sampling of the analytic potential energy landscape; see\cite{TTKAlkanes} for the code to generate the sampling and Figure~\ref{fig:EL-butane-pentane} for the visualization.

For hexane, the landscape is a copy of that for pentane at planar slices parallel to axis planes.
The $1$-cells still occur along a single varying dihedral angle, and the $2$-cells similarly occur along two dihedral angles varying with one dihedral angle fixed.
We now also have $3$-cells, which again give regions where the integral curves directly connect a minimum with a maximum without stopping at a saddle point.

\end{document}